\documentclass[a4paper,12pt]{article}
\usepackage[utf8]{inputenc}

\usepackage{test}

\usepackage{hyperref}
\hypersetup{colorlinks,citecolor=blue,urlcolor=magenta}
\usepackage{doi}

\usepackage[style=authoryear-comp,%numeric-comp,
sorting=nyt,
dashed=false, % don't use dash for repeated author
maxcitenames=2, % maximum names to cite in text
maxbibnames=99, % maximum names to list in reference
uniquelist=false,
uniquename=false,
giveninits=true, % first name initials
natbib, % use natbib commands
date=year % year only
]{biblatex}

\AtBeginRefsection{\GenRefcontextData{sorting=ynt}}
\AtEveryCite{\localrefcontext[sorting=ynt]}

\DeclareFieldFormat{pages}{#1} % suppress pp.
\renewbibmacro{in:}{\ifentrytype{article}{}{\printtext{\bibstring{in}\intitlepunct}}} % delete In:

\usepackage[inline,shortlabels]{enumitem}
\setlist[enumerate,1]{label=(\roman*)}

\usepackage[margin = 1.25in]{geometry}
\usepackage{setspace}
\newcommand{\cA}{\mathcal{A}}
\newcommand{\cP}{\mathcal{P}}

\title{Bubble Necessity Theorem\thanks{We thank Gadi Barlevy, Gaetano Bloise, 
\'Emilien Gouin-Bonenfant, Nobuhiro Kiyotaki, Keiichiro Kobayashi, Narayana Kocherlakota, Erzo Luttmer, 
Jianjun Miao, Herakles Polemarchakis, Guillaume Rocheteau, Jos\'e Scheinkman, Joseph Stiglitz, the editor (Andrew Atkeson), and three anonymous referees for valuable comments and feedback.}}
\author{Tomohiro Hirano\thanks{Department of Economics, Royal Holloway, University of London. %and and Research Associate at the Center for Macroeconomics at the London School of Economics and The Canon Institute for Global Studies. Email: 
\href{mailto:tomohiro.hirano@rhul.ac.uk}{tomohih@gmail.com}.}
\and Alexis Akira Toda\thanks{Department of Economics, Emory University. Email: \href{mailto:alexis.akira.toda@gmail.com}{alexis.akira.toda@gmail.com}.}}
\date{Accepted at \emph{Journal of Political Economy}}

\numberwithin{equation}{section}
\numberwithin{lem}{section}
\numberwithin{prop}{section}
\numberwithin{asmp}{section}

\begin{document}

\maketitle

\begin{abstract}

Asset price bubbles are situations where asset prices exceed the fundamental values defined by the present value of dividends. This paper presents a conceptually new perspective: the \emph{necessity} of bubbles. We establish the Bubble Necessity Theorem in a plausible general class of economic models: with faster long-run economic growth ($G$) than dividend growth ($G_d$) and counterfactual long-run autarky interest rate ($R$) below dividend growth, \emph{all} equilibria are bubbly with non-negligible bubble sizes relative to the economy. This bubble necessity condition naturally arises in economies with sufficiently strong savings motives and multiple factors or sectors with uneven productivity growth.

\medskip


\medskip

\end{abstract}

\section{Introduction}

A rational asset price bubble is a situation in which the asset price ($P$) exceeds its fundamental value ($V$) defined by the present value of dividends ($D$). This paper asks whether asset price bubbles \emph{must} arise, that is, the necessity of bubbles. This question is of fundamental importance. Economists have long held the view that bubbles are either not possible in rational equilibrium models or even if they are, a situation in which asset price bubbles occur is a special circumstance and hence fragile.\footnote{This view is summarized well by the abstract of \citet{SantosWoodford1997}: ``Our main results are concerned with nonexistence of asset pricing bubbles in those economies. These results imply that the conditions under which bubbles are possible---including some well-known examples of monetary equilibria---are relatively fragile.''} If bubbles are inevitable, it would challenge the conventional wisdom and economic modeling. In this paper, we establish a theorem showing that there is a general plausible class of economic models with dividend-paying assets in which asset price bubbles arise in \emph{all} equilibria, implying the necessity of bubbles.

Our question of whether asset price bubbles \emph{must} arise (the necessity of bubbles) is conceptually different from whether asset price bubbles \emph{can} arise (the possibility of bubbles). As is well known from the literature on rational bubbles that starts with \citet{Samuelson1958}, the answer to the latter question is affirmative.\footnote{The development in the literature was mainly theoretical until the 2008 financial crisis. Since then, there has been significant progress in analyzing the relationship between financial frictions and bubbles, as well as in policy and quantitative analyses in light of those developments. For reviews of the literature of rational bubbles, see \citet{Miao2014JME}, \citet{MartinVentura2018}, and \citet{HiranoToda2024JME}.} To date, this literature has almost exclusively focused on pure bubbles, namely assets that pay no dividends ($D=0$) and hence are intrinsically worthless ($V=0$). In these models, there always exists an equilibrium in which the asset price equals its fundamental value (\emph{fundamental equilibrium}), which is zero, as well as a bubbly steady state. In many models, there also exist a continuum of bubbly equilibria converging to the fundamental steady state (asymptotically bubbleless equilibria). Therefore bubbles are possible but not inevitable.

We present a conceptually new perspective on thinking about asset price bubbles: their \emph{necessity}. The main result of our paper can be summarized as follows. Let $G$ be the long-run economic growth rate, $G_d$ the long-run dividend growth rate, and $R$ the counterfactual long-run autarky interest rate. When
\begin{equation}
    R<G_d<G \label{eq:necessity}
\end{equation}
holds, we prove that all equilibria feature asset price bubbles with non-negligible bubble sizes relative to the economy (\emph{asymptotically bubbly equilibria}). The intuition for this result is straightforward. If a fundamental equilibrium exists, in the long run the asset price (the present value of dividends) must grow at the same rate of $G_d$. Then the asset price becomes negligible relative to endowments because $G_d<G$ and the equilibrium consumption allocation approaches autarky. With an autarky interest rate $R<G_d$, the present value of dividends (the fundamental value of the asset) becomes infinite, which is of course impossible in equilibrium. Therefore, there exist no fundamental equilibria nor bubbly equilibria that become asymptotically bubbleless, and all equilibria must be asymptotically bubbly.

We emphasize that the bubble necessity condition \eqref{eq:necessity} naturally arises in plausible economic models. Regarding the inequality $R<G_d$, note that $R$ is the \emph{counterfactual autarky} interest rate, not the \emph{actual equilibrium} interest rate. In models with sufficiently strong savings motives (such as stagnant life cycle income profiles, incomplete markets, or financial frictions), it is not difficult to create a low interest rate environment in the absence of trade. To clearly and convincingly show that the other inequality $G_d<G$ is also natural, in \S\ref{sec:example} we present two example economies in closed-form in which land and stock price bubbles occur as the unique equilibrium outcome. In one example, the economy features two sectors with uneven productivity growth. In one sector, labor (human capital) is the primary input for production such as the labor- or knowledge-intensive sectors. In the other, labor and land are inputs such as the land-intensive agricultural or real estate sectors. We show that when the productivity growth rate is lower in the land-intensive sector, then the land price necessarily exhibits a bubble. In another example, competitive firms produce the consumption good using capital and labor as inputs. We show that when the elasticity of substitution between capital and labor is less than one and capital grows fast, for instance, due to firm creation and/or capital-augmenting technological progress, then the price of capital (the stock price) necessarily exhibits a bubble.

Motivated by the examples in \S\ref{sec:example}, we establish the Bubble Necessity Theorem using workhorse models in macro-finance. In \S\ref{sec:necessity}, we provide it in a classical two-period overlapping generations (OLG) endowment economy under minimal assumptions on preferences, endowments, and dividends. In \S\ref{sec:robustness}, we establish the Theorem in the \citet{Diamond1965} OLG model with capital accumulation and infinite-horizon heterogeneous-agent models in the tradition of \citet{Bewley1977}, one with idiosyncratic investment shocks and another with idiosyncratic preference shocks. We show that under some conditions on technological innovations that enhance the overall productivity, the conditions $R<G_d$ and $G_d<G$ are simultaneously satisfied, necessarily generating asset price bubbles. We emphasize that our result of the necessity of bubbles is not merely a theoretical curiosity but economically relevant and can naturally arise in modern macro-finance models.

\paragraph{Related literature}

The theoretical possibility of asset price bubbles in overlapping generations (OLG) models is well known since \citet{Samuelson1958}. In an infinite-horizon two-agent model with fluctuating income, \citet{Bewley1980} shows that an intrinsically useless asset like fiat money can have a positive value if agents are subject to shortsales constraints. \citet{ScheinkmanWeiss1986} extend \citet{Bewley1980}'s model to continuous time with endogenous labor supply and prove the existence of a recursive monetary equilibrium. Following these seminal work, the theoretical literature has studied both necessary and sufficient conditions for the existence of bubbles.

Regarding sufficient conditions, \citet{OkunoZilcha1983} study an OLG model with time-invariant endowments and preferences and show the existence of a Pareto efficient steady state, which is a competitive equilibrium with or without valued fiat money. An immediate corollary is that if the autarky allocation is Pareto inefficient, then there exists a bubbly (monetary) equilibrium. \citet{AiyagariPeled1991} extend this result to a Markov setting (with potentially a linear storage technology) and show that a stationary allocation is Pareto efficient if and only if the matrix of Arrow prices has spectral radius at most 1, as well as its existence. See \citet{BarbieHillebrand2018} and \citet{BloiseCitanna2019} for recent extensions. This literature shows that bubbles tend to arise when there are gains from trade (the autarky allocation is inefficient).

Regarding necessary conditions, \citet{Kocherlakota1992} considers a deterministic economy with infinitely-lived agents subject to shortsales constraints and proves in Proposition 4 that in any bubbly equilibrium, the present value of the aggregate endowment must be infinite. \citet[Theorem 3.3, Corollary 3.4]{SantosWoodford1997} significantly extend this result to an abstract general equilibrium model with incomplete markets and prove the impossibility of bubbles when dividends are non-negligible relative to the aggregate endowment.\footnote{See \citet[\S3.4]{HiranoToda2024JME} for a simple illustration of this result.} These results imply that there is a fundamental difficulty in generating bubbles attached to dividend-paying assets: bubbles can exist in infinite-horizon models only under sufficient financial frictions that prevent agents from capitalizing infinite wealth and some conditions on dividends. In particular,
\begin{enumerate*}
    \item bubbles are nearly impossible in representative-agent models because market clearing forces the agent to hold the entire asset in equilibrium \citep{Kamihigashi1998,MontrucchioPrivileggi2001}, and
    \item bubbles attached to dividend-paying assets are impossible in models with a steady state.
\end{enumerate*}

A crucial difference of our paper from this literature on the possibility of bubbles is that we prove their necessity. In this respect, our paper is related to \citet{Wilson1981} and \citet{Tirole1985}. \citet{Wilson1981} studies an abstract general equilibrium model with infinitely many agents and commodities. Under standard assumptions, he proves in Theorem 1 that an equilibrium with transfer payments exists and that the transfers can be set to zero (so budgets balance exactly) for agents endowed with only finitely many commodities. However, he shows in \S7 through an example that an equilibrium without transfer payments may not exist and notes that an equilibrium with positive transfers could be interpreted as an equilibrium with money. To our knowledge, this is the first example of the nonexistence of fundamental equilibria. Proposition 1(c) of \citet{Tirole1985} recognizes the possibility that bubbles are necessary for equilibrium existence if the interest rate without bubbles is negative in an OLG model with positive population growth and constant rents, which corresponds to $R<G_d=1<G$. Although he provides some intuition, he did not necessarily provide a formal proof.\footnote{See the discussion in \S\ref{subsec:diamond} and \citet[\S5.2]{HiranoToda2024JME} for details.} Relative to these works, our contribution is that we provide a general nonexistence theorem in workhorse models, and the nonexistence applies not only to fundamental equilibria but also to bubbly equilibria that are asymptotically bubbleless.

Pure bubble models ($D=0$) can naturally be interpreted as models of money. In this context, bubbly and fundamental equilibria are synonymous to monetary and non-monetary equilibria. As discussed in the introduction, these models often admit a continuum of equilibria \citep[Theorem 4]{Gale1973}. \citet{Wallace1980} refers to the fact that the monetary steady state is only one point among a continuum as \emph{tenuous}. \citet{Scheinkman1978} introduces arbitrarily small dividends to the pure bubble asset to eliminate the non-monetary steady state with $P=0$, which is sometimes called commodity money refinement; see also \citet[\S2.1]{Brock1990}. However, commodity money refinement only rules out equilibria converging to the non-monetary steady state,\footnote{\citet*[p.~411]{LagosRocheteauWright2017} state ``While [commodity money refinement] rules out equilibria [converging to the non-monetary steady state], there can still exist cyclic, chaotic and stochastic equilibria''.} and it does not consider whether the equilibrium with positive dividends is bubbly. \citet[Theorem 4]{BrockScheinkman1980} and \citet{Scheinkman1980} provide sufficient conditions on preferences and endowments (for instance, $\lim_{x\to 0}xu'(x)>0$ and the old have no endowment) for the nonexistence of monetary equilibria converging to the non-monetary steady state, which \citet{Santos1990} further generalizes. \citet{BalaskoShell1980,BalaskoShell1981_2} study the existence of equilibria in OLG models with and without money. With Cobb-Douglas utility functions, \citet[Proposition 3.3]{BalaskoShell1981_3} show that the equilibrium is unique without money and is one-dimensional (in particular, a continuum) with money. With the gross substitute property, \citet*[Theorem G]{KehoeLevineMas-CollelWoodford1991} show the uniqueness of monetary and non-monetary steady states. Our results are different from these monetary models because we prove the nonexistence of any fundamental equilibria and any bubbly equilibria that are asymptotically bubbleless in economies with dividend-paying assets under minimal assumptions on preferences and endowments.

\section{Definition and characterization of bubbles}\label{sec:prelim}

We define bubbles following the literature. (See \citet[Ch.~5]{BlanchardFischer1989} and \citet[\S 13.6]{Miao2014Dynamics} for textbook treatment.) We consider an infinite-horizon, deterministic economy with a homogeneous good and time indexed by $t=0,1,\dotsc$. Consider an asset with infinite maturity that pays dividend $D_t\ge 0$ and trades at ex-dividend price $P_t$, both in units of the time-$t$ good. In the background, we assume the presence of rational, perfectly competitive investors. Free disposal of the asset implies $P_t\ge 0$.\footnote{\label{fn:P>0}If $P_t<0$, by purchasing one additional share of the asset at time $t$ and immediately disposing of it, an investor can increase consumption at time $t$ by $-P_t>0$ with no cost, which violates individual optimality.} Let $q_t>0$ be the Arrow-Debreu price, \ie, the date-0 price of the consumption good delivered at time $t$, with the normalization $q_0=1$. The absence of arbitrage implies
\begin{equation}
    q_tP_t = q_{t+1}(P_{t+1}+D_{t+1}). \label{eq:noarbitrage}
\end{equation}
Iterating the no-arbitrage condition \eqref{eq:noarbitrage} forward and using $q_0=1$, we obtain
\begin{equation}
    P_0=\sum_{t=1}^T q_tD_t+q_TP_T. \label{eq:P_iter}
\end{equation}
Noting that $P_t\ge 0$, $D_t\ge 0$, and $q_t>0$, the infinite sum of the present value of dividends $V_0\coloneqq \sum_{t=1}^\infty q_tD_t$ exists, which is called the \emph{fundamental value} of the asset. Letting $T\to\infty$ in \eqref{eq:P_iter}, we obtain
\begin{equation}
    P_0=\sum_{t=1}^\infty q_tD_t+\lim_{T\to\infty}q_TP_T=V_0+\lim_{T\to\infty}q_TP_T. \label{eq:P0}
\end{equation}
We say that the \emph{transversality condition} for asset pricing holds if
\begin{equation}
    \lim_{T\to\infty}q_TP_T = 0, \label{eq:TVC}
\end{equation}
in which case the identity \eqref{eq:P0} implies that $P_0=V_0$ and the asset price equals its fundamental value. If 
the transversality condition \eqref{eq:TVC} is violated, or $\lim_{T\to\infty}q_TP_T>0$, then $P_0>V_0$, and we say that the asset contains a \emph{bubble}.

Note that in deterministic economies, for all $t$ we have
\begin{equation}
    P_t=\underbrace{\frac{1}{q_t}\sum_{s=1}^\infty q_{t+s}D_{t+s}}_\text{fundamental value $V_t$}+\underbrace{\frac{1}{q_t}\lim_{T\to\infty} q_TP_T}_\text{bubble component}. \label{eq:Pt}
\end{equation}
Therefore either $P_t=V_t$ for all $t$ or $P_t>V_t$ for all $t$, so the economy is permanently in either the bubbly or the fundamental regime. Thus, a bubble is a permanent overvaluation of an asset, which is a feature of rational expectations.

In general, checking the transversality condition \eqref{eq:TVC} directly could be difficult because it involves $q_T$. The following lemma, which is Proposition 7 of \citet{Montrucchio2004}, provides an equivalent characterization. This lemma, although quite simple, may be of independent interest because it significantly facilitates checking the presence or absence of bubbles. We provide a proof in the main text as it is short and the lemma is useful for the subsequent proofs.

\begin{lem}[Bubble Characterization, \citealp{Montrucchio2004}]\label{lem:bubble}
If $P_t>0$ for all $t$, the asset price exhibits a bubble if and only if $\sum_{t=1}^\infty D_t/P_t<\infty$.
\end{lem}

\begin{proof}
Changing $t$ to $t-1$ in the no-arbitrage condition \eqref{eq:noarbitrage} and dividing both sides by $q_tP_t>0$, we obtain $q_{t-1}P_{t-1}/q_tP_t=1+D_t/P_t$. Multiplying from $t=1$ to $t=T$, expanding terms, and using $1+x\le \e^x$, we obtain
\begin{equation*}
    1+\sum_{t=1}^T\frac{D_t}{P_t}\le \frac{q_0P_0}{q_TP_T}=\prod_{t=1}^T\left(1+\frac{D_t}{P_t}\right)\le \exp\left(\sum_{t=1}^T\frac{D_t}{P_t}\right).
\end{equation*}
Letting $T\to \infty$, we have $\lim_{T\to\infty}q_TP_T>0$ if and only if $\sum_{t=1}^\infty D_t/P_t<\infty$.
\end{proof}

If $D_t>0$ infinitely often, then $P_t>0$ by \eqref{eq:Pt} and the assumption is satisfied. Lemma \ref{lem:bubble} implies that there is an asset price bubble if and only if the infinite sum of dividend yields $D_t/P_t$ is finite. Because $\sum_{t=1}^\infty 1/t=\infty$ but $\sum_{t=1}^\infty 1/t^\alpha<\infty$ for any $\alpha>1$, roughly speaking, there is an asset price bubble if the price-dividend ratio $P_t/D_t$ grows faster than linearly.

\section{Bubble necessity in example economies}\label{sec:example}

In this section, to convince the reader of the necessity of asset price bubbles in plausible economic models, we present example economies with unique equilibria in which asset price bubbles necessarily arise. Throughout this section, for tractability we consider two-period overlapping generations (OLG) economies with Cobb-Douglas utility function $U(y,z)=(1-\beta)\log y+\beta\log z$, where $y,z$ denote the consumption when young and old and $\beta\in (0,1)$. The population size is constant at 1 unless otherwise stated.

\subsection{Textbook endowment economy}\label{subsec:example_textbook}

We first consider a textbook endowment economy. The initial old are endowed with a unit supply of an asset with infinite maturity. At time $t$, the young are endowed with $a_t>0$ units of the consumption good, the old none, and the asset pays dividend $D_t>0$. A competitive equilibrium with sequential trading is defined by a sequence $\set{(P_t,x_t,y_t,z_t)}_{t=0}^\infty$ of asset price $P_t$, asset holdings of young $x_t$, and consumption of young and old $(y_t,z_t)$ such that
\begin{enumerate*}
    \item the young maximize utility subject to the budget constraints $y_t+P_tx_t=a_t$ and $z_{t+1}=(P_{t+1}+D_{t+1})x_t$,
    \item commodity market clears: $y_t+z_t=a_t+D_t$, and
    \item asset market clears: $x_t=1$.
\end{enumerate*}
The following proposition provides a necessary and sufficient condition for bubbles.

\begin{prop}\label{prop:textbook}
There exists a unique equilibrium, and the asset price exhibits a bubble if and only if $\sum_{t=1}^\infty D_t/a_t<\infty$.
\end{prop}

\begin{proof}
Due to log utility, the optimal consumption of the young is $y_t=(1-\beta)a_t$. Asset market clearing and the budget constraint of the young imply $P_t=P_tx_t=a_t-y_t=\beta a_t$. Clearly, the equilibrium is unique. Since the dividend yield is $D_t/P_t=D_t/(\beta a_t)$, the claim follows from Lemma \ref{lem:bubble}.
\end{proof}

\subsection{Two-sector growth economy with land}\label{subsec:example_twosector}

The endowment economy in \S\ref{subsec:example_textbook} is arguably stylized. We next provide a micro-foundation by considering a production economy with labor and land.

The initial old are endowed with a unit supply of land, which is durable and non-reproducible. Each period, the young are endowed with one unit of labor and the old none. There are two production technologies with time $t$ production functions given by
\begin{subequations}
    \begin{align}
        F_{1t}(H,X)&=G_1^tH, \label{eq:F1}\\
        F_{2t}(H,X)&=G_2^tH^\alpha X^{1-\alpha}, \label{eq:F2}
    \end{align}
\end{subequations}
where $H,X$ denote the inputs of labor and land and $\alpha\in (0,1)$. Sector 1 uses labor as the primary input and can be interpreted as a labor-intensive industry with total productivity growth rate $G_1$ such as service or knowledge-intensive sectors. Sector 2 uses both labor and land as inputs and can be interpreted as a land-intensive sector such as agriculture or real estate with a Cobb-Douglas production function and productivity growth rate $G_2$.

A competitive equilibrium with sequential trading is defined by a sequence
\begin{equation*}
    \set{(P_t,r_t,w_t,x_t,y_t,z_t,H_{1t},H_{2t})}_{t=0}^\infty
\end{equation*}
of land price $P_t$, land rent $r_t$, wage $w_t$, land holdings of young $x_t$, consumption of young and old $(y_t,z_t)$, and labor allocation $(H_{1t},H_{2t})$ such that,
\begin{enumerate*}
    \item the young maximize utility subject to the budget constraints $y_t+P_tx_t=w_t$ and $z_{t+1}=(P_{t+1}+r_{t+1})x_t$,
    \item firms maximize profits,
    \item commodity market clears: $y_t+z_t=G_1^tH_{1t}+G_2^tH_{2t}^\alpha$,
    \item labor market clears: $H_{1t}+H_{2t}=1$, and
    \item land market clears: $x_t=1$.
\end{enumerate*}
In this model, we can prove the necessity of land bubbles under some conditions on productivity growth rates in the two sectors.

\begin{prop}\label{prop:twosector}
If $G_1>G_2$, then the unique equilibrium land price is $P_t=\beta G_1^t$, and there is a bubble.
\end{prop}

\begin{proof}
In equilibrium, it must be $H_{2t}>0$ by the Inada condition. Therefore profit maximization of Sector 2 implies $w_t=\alpha G_2^tH_{2t}^{\alpha-1}$. If $H_{2t}=1$, then $w_t=\alpha G_2^t<G_1^t$ because $\alpha<1$ and $G_2<G_1$. Then firms in Sector 1 make infinite profits, which is a contradiction. Therefore in equilibrium it must be $H_{2t}\in (0,1)$, and profit maximization of Sector 1 implies $w_t=G_1^t$. Therefore labor in Sector 2 satisfies
\begin{equation*}
    \alpha G_2^tH_{2t}^{\alpha-1}=w_t=G_1^t\iff H_{2t}=\alpha^\frac{1}{1-\alpha}(G_2/G_1)^\frac{t}{1-\alpha}.
\end{equation*}
By profit maximization in Sector 2 and $X=1$, land rent satisfies
\begin{equation*}
    r_t=(1-\alpha)G_2^tH_{2t}^\alpha=(1-\alpha)\alpha^\frac{\alpha}{1-\alpha}G_2^t(G_2/G_1)^\frac{\alpha t}{1-\alpha}.
\end{equation*}
As in Proposition \ref{prop:textbook}, the young consume $y_t=(1-\beta)w_t$, the land price equals $P_t=\beta w_t=\beta G_1^t$, and hence the equilibrium is unique. The dividend yield on land is then
\begin{equation*}
    \frac{r_t}{P_t}=\frac{(1-\alpha)\alpha^\frac{\alpha}{1-\alpha}G_2^t(G_2/G_1)^\frac{\alpha t}{1-\alpha}}{\beta G_1^t}=\frac{(1-\alpha)\alpha^\frac{\alpha}{1-\alpha}}{\beta}(G_2/G_1)^\frac{t}{1-\alpha},
\end{equation*}
which decays geometrically and is summable because $G_2<G_1$ by assumption. Lemma \ref{lem:bubble} implies a land bubble.
\end{proof}

In this example, both the labor income and land price grow at rate $G\coloneqq G_1$, while the land rent grows at rate $G_d \coloneqq G_2(G_2/G_1)^\frac{\alpha}{1-\alpha}$. The condition $G_1>G_2$ in Proposition \ref{prop:twosector} is equivalent to the bubble necessity condition $G>G_d$ discussed in the introduction.\footnote{There is empirical support for $G_1>G_2$. According to \citet[p.~698, Figure 20.1]{Acemoglu2009}, the employment share of agriculture in U.S. has declined from about 80\% to below 5\% over the past two centuries, so the technological growth rate in the ``land'' sector has been lower than the whole economy.}

\subsection{Innovation and stock market bubble}\label{subsec:example_innovation}

In the model of \S\ref{subsec:example_twosector}, the bubble was attached to the sector with slower growth (land-intensive sector). However, bubbles can also arise in a sector or a production factor with faster growth such as the stock market. To illustrate this point, we consider a production economy with capital and labor.

Each period, firms produce the consumption good using the neoclassical production function $F(K,L)$, where $K,L$ denote capital and labor inputs. To simplify the analysis, we exogenously specify aggregate capital and labor at time $t$ as $K_t,L_t$.\footnote{\label{fn:endogenous growth}Because the equilibrium conditions are the same regardless of whether $K_t,L_t$ are exogenous or not, the subsequent argument applies to endogenous growth models as well. This is similar to the fact that Euler equations hold in consumption-based asset pricing models regardless of whether consumption is exogenous or not.} Growth in these variables can be interpreted in various ways: it could be growth in quantity (firm creation or population growth), quality (factor-augmenting technological progress), or a combination thereof. The initial old are endowed with the stock market index (claims to aggregate capital rents) with outstanding shares normalized to 1. Only the young are endowed with labor.

Letting $r_t,w_t$ be the rental and wage rates, profit maximization implies $r_t=F_K(K_t,L_t)$ and $w_t=F_L(K_t,L_t)$. As before, the asset price (stock market index $P_t$) equals aggregate savings, which equals fraction $\beta$ of aggregate labor income: $P_t=\beta w_tL_t$. Dividend equals aggregate rents: $D_t=r_tK_t$. Putting all the pieces together, the dividend yield is given by
\begin{equation}
    \frac{D_t}{P_t}=\frac{1}{\beta}\frac{F_K(K_t,L_t)K_t}{F_L(K_t,L_t)L_t}. \label{eq:divyield}
\end{equation}

For concreteness, suppose capital and labor grow at rates $G_K,G_L$ and that the production function takes the constant elasticity of substitution (CES) form
\begin{equation*}
    F(K,L)=\left(\alpha K^{1-1/\sigma}+(1-\alpha) L^{1-1/\sigma}\right)^\frac{1}{1-1/\sigma},
\end{equation*}
where $\sigma>0$ is the elasticity of substitution and $\alpha\in (0,1)$. Therefore we obtain the following proposition.

\begin{prop}\label{prop:CES}
There is a stock market bubble if and only if
\begin{equation}
    (\sigma-1)(G_K-G_L)<0. \label{eq:CES_bubble}
\end{equation}
\end{prop}

\begin{proof}
Under the maintained assumptions, the dividend yield \eqref{eq:divyield} becomes
\begin{equation*}
    \frac{D_t}{P_t}=\frac{\alpha}{\beta(1-\alpha)}\left((G_K/G_L)^t(K_0/L_0)\right)^{1-1/\sigma},
\end{equation*}
which decays geometrically if and only if either
\begin{enumerate*}
    \item $\sigma<1$ and $G_K>G_L$ or
    \item $\sigma>1$ and $G_K<G_L$.
\end{enumerate*}
These conditions are equivalent to \eqref{eq:CES_bubble}. Otherwise, $D_t/P_t$ is bounded below by a positive constant. Therefore the claim is immediate from Lemma \ref{lem:bubble}.
\end{proof}

Empirical evidence suggests $\sigma<1$.\footnote{See \citet{OberfieldRaval2021} for a study using micro data and \citet*{GechertHavranekIrsovaKolcunova2022} for a literature review and metaanalysis.} Therefore if $G_K>G_L$, so capital grows at a faster rate than labor, for instance, due to innovation or firm creation, the dividend yield geometrically decays and a stock market bubble emerges as the unique equilibrium outcome. If the reverse inequality of \eqref{eq:CES_bubble} strictly holds, the stock market index return diverges to $\infty$ and the price-dividend ratio converges to 0, which is counterfactual. If $\sigma=1$ or $G_K=G_L$, the dividend yield is constant but this is obviously a knife-edge case.

\section{Bubble necessity in OLG endowment economies}\label{sec:necessity}

In \S\ref{sec:example}, we showed the necessity of asset price bubbles in several example economies. In this section, we establish the necessity in an abstract two-period OLG model.

\subsection{Model}\label{subsec:model}

We consider an OLG endowment economy with a long-lived asset.

\paragraph{Agents}

At each date $t=0,1,\dotsc$, a unit mass of a new generation of agents are born, who live for two dates. An agent born at time $t$ has utility function $U_t(y_t,z_{t+1})$, where $y_t,z_{t+1}$ denote the consumption when young and old. At $t=0$, there is a unit mass of old agents, who care only about their consumption $z_0$.

\paragraph{Commodities and asset}

There is a single perishable good. The endowments of the young and old at time $t$ are denoted by $a_t>0$ and $b_t\ge 0$. There is a unit supply of a dividend-paying asset with infinite maturity. Let $D_t\ge 0$ be the dividend of the asset at time $t$.

\paragraph{Budget constraints}

Letting $P_t\ge 0$ be the asset price and $x_t$ the number of asset shares demanded by the young, the budget constraints are
\begin{subequations}\label{eq:budget}
\begin{align}
    &\text{Young:} && y_t+P_tx_t=a_t, \label{eq:budget_young}\\
    &\text{Old:} && z_{t+1}=b_{t+1}+(P_{t+1}+D_{t+1})x_t. \label{eq:budget_old}
\end{align}
\end{subequations}
That is, at time $t$ the young decide to spend income $a_t$ on consumption $y_t$ and asset purchase $P_tx_t$, and at time $t+1$ the old liquidate all wealth to consume. The asset demand $x_t$ is arbitrary (positive or negative) as long as consumption is nonnegative.

\paragraph{Equilibrium}

Our equilibrium notion is the competitive equilibrium with sequential trading. A \emph{competitive equilibrium} consists of a sequence of prices and allocations $\set{(P_t,x_t,y_t,z_t)}_{t=0}^\infty$ satisfying the following conditions.
\begin{enumerate}
    \item (Individual optimization) The initial old consume $z_0=b_0+P_0+D_0$; for all $t$, the young maximize utility $U_t(y_t,z_{t+1})$ subject to the budget constraints \eqref{eq:budget}.
    \item (Commodity market clearing) $y_t+z_t=a_t+b_t+D_t$ for all $t$.
    \item (Asset market clearing) $x_t=1$ for all $t$.
\end{enumerate}

Note that, because at each date there are only two types of agents (the young and old) and the old exit the economy, market clearing forces that the young buy the entire shares of the asset, which explains $x_t=1$. Consequently, using the budget constraint \eqref{eq:budget}, the equilibrium allocation is
\begin{equation}
    (y_t,z_t)=(a_t-P_t,b_t+P_t+D_t). \label{eq:eqallocation}
\end{equation}

We define fundamental and bubbly equilibria in the obvious way.

\begin{defn}\label{defn:bubble}
An equilibrium is \emph{fundamental} (\emph{bubbly}) if $P_0=V_0$ ($P_0>V_0$).
\end{defn}

By definition, a bubbly equilibrium is an equilibrium in which the asset price exceeds its fundamental value. However, pure bubble models with $D_t=0$ often admit fundamental and bubbly steady states as well as a continuum of bubbly equilibria converging to the fundamental steady state.\footnote{See \citet[\S3.1]{HiranoToda2024JME} for a detailed analysis of a specific pure bubble model.} Therefore, proving the nonexistence of fundamental equilibria alone may not be convincing because it does not rule out bubbly equilibria converging to the fundamental steady state. We thus define a more demanding equilibrium concept, \emph{asymptotically bubbly equilibria}, which are bubbly equilibria with non-negligible bubble sizes relative to the economy.

\begin{defn}[Asymptotically bubbly equilibria]\label{defn:asym_bubble}
Let $\set{P_t}_{t=0}^\infty$ be equilibrium asset prices. The asset is \emph{asymptotically relevant (irrelevant)} if
\begin{equation}
    \liminf_{t\to\infty}\frac{P_t}{a_t}>0\quad (=0). \label{eq:relevant}
\end{equation}
A bubbly equilibrium is \emph{asymptotically bubbly (bubbleless)} if the asset is asymptotically relevant (irrelevant).
\end{defn}

By Definition \ref{defn:asym_bubble}, the set of asymptotically bubbly equilibria is a subset of bubbly equilibria. The relevance condition \eqref{eq:relevant} implies that there exists $p>0$ such that $P_t/a_t\ge p$ for all large enough $t$, so the asset price is always a non-negligible fraction of the endowment of the young. Because the young are the natural buyer of the asset, this condition implies that there will be trade in the asset in the long run, which motivates the term ``relevance''.\footnote{In monetary theory, money is said to be \emph{essential} if trading money is necessary for achieving at least some (desirable) allocations \citep{Hahn1973,Wallace2001}. When condition \eqref{eq:relevant} holds, the equilibrium allocation \eqref{eq:eqallocation} is relatively far from the autarky allocation $(a_t,b_t)$ because the relative error in the consumption of the young, $P_t/a_t$, is bounded away from zero. In the sense that asset trading achieves a nontrivial equilibrium allocation, the relevance condition \eqref{eq:relevant} is similar to essentiality. Obviously, the essentiality of money is a completely different concept from the necessity of bubbles.}

\subsection{Existence and characterization of equilibrium}

We introduce the following assumptions.

\begin{asmp}\label{asmp:U}
For all $t$, the utility function $U_t:\R_+^2\to [-\infty,\infty)$ is continuous, quasi-concave, and continuously differentiable on $\R_{++}^2$ with positive partial derivatives.
\end{asmp}

Assumption \ref{asmp:U} is standard. We define the marginal rate of substitution (MRS)
\begin{equation}
    M_t(y,z)\coloneqq \frac{(U_t)_z(y,z)}{(U_t)_y(y,z)}>0,\label{eq:MRS}
\end{equation}
where we denote partial derivatives with subscripts, \eg, $U_y=\partial U/\partial y$. The following theorem shows that an equilibrium always exists under the maintained assumptions. Furthermore, it characterizes the equilibrium as a solution to a nonlinear difference equation.

\begin{thm}\label{thm:exist}
Suppose Assumption \ref{asmp:U} holds and let $M_t$ be the marginal rate of substitution in \eqref{eq:MRS}. Then an equilibrium exists and asset prices satisfy
\begin{equation}
    P_t=\min\set{M_t(y_t,z_{t+1})(P_{t+1}+D_{t+1}),a_t},\label{eq:foc}
\end{equation}
where $(y_t,z_{t+1})=(a_t-P_t,b_{t+1}+P_{t+1}+D_{t+1})$.
\end{thm}

Note that if $P_t<a_t$, \eqref{eq:foc} reduces to the familiar asset pricing equation
\begin{equation*}
    P_t=M_t(y_t,z_{t+1})(P_{t+1}+D_{t+1}).
\end{equation*}
In general, the min operator appears in \eqref{eq:foc} due to the nonnegativity constraint on consumption: $y_t=a_t-P_tx_t\ge 0$. If we assume the Inada condition $U_y(0,z)=\infty$, this constraint never binds. If $0<P_t<a_t$, because the economy is deterministic, we can define the interest rate between time $t$ and $t+1$ by
\begin{equation}
    R_t\coloneqq \frac{P_{t+1}+D_{t+1}}{P_t}=\frac{1}{M(y_t,z_{t+1})}. \label{eq:Rt}
\end{equation}

\subsection{Bubble Necessity Theorem}

The question of fundamental importance is whether bubbles must arise in equilibrium, that is, whether it is possible that fundamental equilibria or asymptotically bubbleless equilibria may fail to exist. We address this question under additional assumptions.

\begin{asmp}\label{asmp:ab}
The endowments $\set{(a_t,b_t)}_{t=0}^\infty\subset \R_{++}\times \R_+$ satisfy
\begin{subequations}\label{eq:ab}
    \begin{align}
        \lim_{t\to\infty}\frac{a_{t+1}}{a_t}&\eqqcolon G \in (0,\infty), \label{eq:G}\\
        \lim_{t\to\infty}\frac{b_t}{a_t}&\eqqcolon w\in [0,\infty). \label{eq:w}
    \end{align}
\end{subequations}
\end{asmp}

Condition \eqref{eq:G} implies that the endowments of the young grow at rate $G>0$ in the long run. Although we use the term ``growth'', it could be $G\le 1$, so stationary or shrinking economies are also allowed. Furthermore, \eqref{eq:G} is an assumption only at infinity, so the endowments are arbitrary for arbitrarily long finite periods. Condition \eqref{eq:w} implies that the old-to-young endowment ratio approaches $w$ in the long run. Again, the income ratio is arbitrary for arbitrarily long finite periods.

Motivated by the interest rate formula \eqref{eq:Rt}, we define
\begin{equation}
    f_t(y,z)\coloneqq \frac{1}{M_t(a_ty,a_tz)}=\frac{(U_t)_y(a_ty,a_tz)}{(U_t)_z(a_ty,a_tz)}, \label{eq:ft}
\end{equation}
which gives the forward rate between time $t$ and $t+1$ when generation $t$ consumes $(a_ty,a_tz)$. We impose the following uniform convergence condition on $f_t$.

\begin{asmp}\label{asmp:f}
The forward rate function $f_t$ is continuous and uniformly converges on compact sets: there exists a continuous function $f:\R_{++}\times \R_+\to \R_+$ such that for any nonempty compact set $K\subset \R_{++}\times \R_+$, we have
\begin{equation}
    \lim_{t\to\infty}\sup_{(y,z)\in K}\abs{f_t(y,z)-f(y,z)}=0. \label{eq:f}
\end{equation}
\end{asmp}

A few remarks are in order. First, the domain of $f_t$ and $f$ is $\R_{++}\times \R_+$, not $\R_{++}^2$. This is because we would like to allow the possibility of zero endowments for the old, as in the examples in \S\ref{sec:example}. For instance, suppose the utility function exhibits constant relative risk aversion (CRRA), so $U_t(y,z)=u(y)+\beta u(z)$ with
\begin{equation}
    u(c)=\begin{cases*}
        \frac{c^{1-\gamma}}{1-\gamma} & if $0<\gamma\neq 1$,\\
        \log c & if $\gamma=1$,
    \end{cases*}\label{eq:CRRA}
\end{equation}
where $\beta>0$ is the discount factor and $\gamma>0$ is the relative risk aversion coefficient. Then $f_t(y,z)=f(y,z)=(z/y)^\gamma/\beta$, which satisfies Assumption \ref{asmp:f}. More generally, $f_t$ is well defined and continuous on $\R_{++}\times \R_+$ if $\partial U_t/\partial z$ can be continuously extended on $\R_{++}\times \R_+$ by allowing the value $\infty$.

Second, there are two interpretations of Assumption \ref{asmp:f}. If the utility function $U_t$ is homothetic, then the marginal rate of substitution \eqref{eq:MRS} is homogeneous of degree 0, so $f_t(y,z)=1/M_t(y,z)$ by \eqref{eq:ft}. Alternatively, if there is population growth as in \citet{Tirole1985}, $a_t$ denotes the population size, and $u_t(y,z)\coloneqq U_t(a_ty,a_tz)$ denotes the utility of an agent consuming $(y,z)$, then $f_t$ in \eqref{eq:ft} is exactly the reciprocal of the MRS of $u_t$. In either case, condition \eqref{eq:f} reduces to the uniform convergence of MRS instead of the scaled MRS.

Under the maintained assumptions, we can prove the necessity of bubbles.

\begin{thm}[Necessity of bubbles in OLG model]\label{thm:necessity}
If Assumptions \ref{asmp:U}--\ref{asmp:f} hold and
\begin{equation}
    f(1,Gw)<G_d\coloneqq \limsup_{t\to\infty} D_t^{1/t}<G, \label{eq:necessity_OLG}
\end{equation}
then all equilibria are asymptotically bubbly.
\end{thm}

Note that the statement of Theorem \ref{thm:necessity} is not vacuous because an equilibrium always exists by Theorem \ref{thm:exist}. Although the proof of Theorem \ref{thm:necessity} is technical, its intuition is straightforward. For simplicity, suppose $(a_t,D_t)=(a_0G^t,D_0G_d^t)$ with $G_d<G$, so endowments and dividends grow at constant rates. If a fundamental equilibrium exists, the price $P_t$ must asymptotically grow at rate $G_d$ because (by the definition of a fundamental equilibrium) it must equal the present value of dividends, which grow at rate $G_d$. Because $G_d<G$, the price-income ratio $P_t/a_t$ grows at rate $G_d/G<1$ and hence converges to zero, so the asset becomes asymptotically irrelevant. Hence the budget constraint \eqref{eq:budget} together with the asset market clearing condition $x_t=1$ implies that the consumption allocation $(y_t,z_{t+1})$ approaches autarky in the long run, and the equilibrium interest rate \eqref{eq:Rt} converges to
\begin{equation}
    R_t=\frac{1}{M_t(y_t,z_{t+1})}\to f(1,Gw)<G_d \label{eq:Rbound}
\end{equation}
by Assumption \ref{asmp:f} and condition \eqref{eq:necessity_OLG}. However, \eqref{eq:Rbound} implies that the interest rate is asymptotically lower than the dividend growth rate, so the fundamental value of the asset is infinite, which is of course impossible in equilibrium. Thus a fundamental equilibrium cannot exist.

\subsection{Examples}

Clearly, the examples in \S\ref{subsec:example_twosector}, \ref{subsec:example_innovation} are special cases of Theorem \ref{thm:necessity} (after converting to endowment economies), but they both use log utility for tractability. We present two more examples with other utility functions for illustration.

\begin{exmp}[\citealp{Wilson1981}]\label{exmp:linear}
Let the utility function be $U_t(y,z)=y+\beta z$, endowments $(a_t,b_t)=(aG^t,bG^t)$, and dividends $D_t=DG_d^t$. Since the marginal rate of substitution $M(y,z)=\beta$ is constant, Assumptions \ref{asmp:U}--\ref{asmp:f} clearly hold. The bubble necessity condition \eqref{eq:necessity_OLG} holds if $1/\beta<G_d<G$. The example in \citet[\S 7]{Wilson1981} is a special case with $\beta=3$, $G_d=1/2$, and $G=1$.
\end{exmp}

In Appendix \ref{sec:unique}, we prove the unique equilibrium price in Example \ref{exmp:linear} is $P_t=aG^t$.

\begin{exmp}[CRRA]\label{exmp:CRRA}
Suppose agents have the CRRA utility in the remark after Assumption \ref{asmp:f}, endowments grow at a constant rate $G>1$, so $(a_t,b_t)=(aG^t,bG^t)$, and dividends are positive but constant: $D_t=D>0$. Assumptions \ref{asmp:U}--\ref{asmp:f} clearly hold. Since $G_d=1$, \eqref{eq:necessity_OLG} reduces to
\begin{equation}
    \frac{1}{\beta}(bG/a)^\gamma<1<G\iff a>\beta^{-1/\gamma}Gb. \label{eq:nonexist_CRRA}
\end{equation}
Thus if \eqref{eq:nonexist_CRRA} holds, all assumptions of Theorem \ref{thm:necessity} are satisfied and all equilibria are asymptotically bubbly.
\end{exmp}

The bubble necessity condition \eqref{eq:nonexist_CRRA} implies that, if the young are sufficiently rich, all equilibria are asymptotically bubbly. The intuition is that the young have a strong savings motive, it pushes down the interest rate, but the interest rate cannot fall below the dividend growth rate, for otherwise the asset price would be infinite. This implies that when the economy falls into a low interest rate environment, the only possible outcome is an asset price bubble.

Although an application of Theorem \ref{thm:necessity} to Example \ref{exmp:CRRA} implies that all equilibria are asymptotically bubbly, Theorem \ref{thm:necessity} is silent about how to analyze them. For this purpose, we may apply the Hartman-Grobman theorem \citep[Theorem 4.6]{Chicone2006} (essentially linearization around the steady state); see Appendix \ref{sec:global} for details.

\section{Robustness of bubble necessity}\label{sec:robustness}

\S\ref{sec:necessity} provides sufficient conditions for the necessity of asset price bubbles in overlapping generations endowment economies with exogenous dividends. Although we chose this model for simplicity and clarity, there is nothing special about this apparently restrictive setting. In fact, some examples in \S\ref{sec:example} feature production and endogenous dividends. In this section, to further show the robustness of our results, we consider the \citet{Diamond1965} OLG model with capital accumulation and Bewley-type infinite-horizon heterogeneous-agent models. Because our purpose is to show the necessity of bubbles in plausible economic models, our models make minimal modifications to well-known models in the literature.

Since all examples below concern production economies, we modify the definition of asymptotically bubbly equilibria in Definition \ref{defn:asym_bubble} as follows.

\begin{defn}
Let $\set{P_t}_{t=0}^\infty$ be equilibrium asset prices. A bubbly equilibrium is \emph{asymptotically bubbly} if $\liminf_{t\to\infty}P_t/G^t>0$ for suitably defined economic growth rate $G>0$.
\end{defn}

\subsection{Diamond OLG model with capital accumulation}\label{subsec:diamond}

We first consider the \citet{Diamond1965} OLG model with capital accumulation.

\paragraph{Model}

The model description is brief because it is well known. Agents have Cobb-Douglas utility with discount factor $\beta\in (0,1)$ as in \S\ref{sec:example}. To simplify notation, let $F(K,L)$ be a neoclassical production function including undepreciated capital, so $F$ is homogeneous of degree 1, concave, and continuously differentiable on $\R_{++}^2$ with positive partial derivatives. Let the aggregate labor be $L=1$. There is an asset in unit supply, which pays exogenous dividend $D_t\ge 0$ at time $t$ and trades at endogenous ex-dividend price $P_t$.

\paragraph{Equilibrium}
In any equilibrium, the wage is $w_t=F_L(K_t,1)$. Because the young save fraction $\beta$ of wealth, the budget constraint and the market clearing condition for the asset imply
\begin{equation}
    K_{t+1}+P_t=\beta F_L(K_t,1). \label{eq:budget_diamond}
\end{equation}
The no-arbitrage condition implies
\begin{equation}
    \frac{P_{t+1}+D_{t+1}}{P_t}=R_t=F_K(K_{t+1},1). \label{eq:noarbitrage_diamond}
\end{equation}

We introduce the following technical condition.

\begin{asmp}\label{asmp:FL}
There exists $K^*>0$ such that
\begin{equation}
    \beta F_L(K,1)-K\begin{cases*}
        >0 & if $0<K<K^*$,\\
        =0 & if $K=K^*$,\\
        <0 & if $K>K^*$.
    \end{cases*} \label{eq:K*}
\end{equation}
\end{asmp}

Using the budget constraint \eqref{eq:budget_diamond}, the steady state savings (excluding capital) is $\beta F_L(K,1)-K$. Assumption \ref{asmp:FL} implies that savings is positive (negative) at low (high) capital values. The following theorem shows the necessity of asset price bubbles under the condition $R<G_d<G$.

\begin{thm}[Necessity of bubbles in Diamond model]\label{thm:diamond}
If Assumption \ref{asmp:FL} holds and
\begin{equation}
	F_K(K^*,1)<G_d\coloneqq \limsup_{t\to\infty}D_t^{1/t}<1\eqqcolon G, \label{eq:necessity_diamond}
\end{equation}
then any equilibrium with $\liminf_{t\to\infty}K_t>0$ is asymptotically bubbly.
\end{thm}

Note that the bubble necessity condition \eqref{eq:necessity_diamond} exactly parallels that of Theorem \ref{thm:necessity} in \eqref{eq:necessity_OLG}. Theorem \ref{thm:diamond} is different from Proposition 1(c) of \citet{Tirole1985} because Tirole imposes several high-level assumptions without providing examples satisfying them, assumes a constant dividend growth rate, and does not rule out asymptotically bubbleless equilibria.

\begin{exmp}
Suppose $F$ is Cobb-Douglas with depreciation rate $\delta$, so $F(K,L)=AK^\alpha L^{1-\alpha}+(1-\delta)K$ for $A>0$, $\alpha\in (0,1)$, and $\delta\in [0,1]$. Then
\begin{equation*}
    \beta F_L(K,1)-K=K^\alpha(\beta A(1-\alpha)-K^{1-\alpha}),
\end{equation*}
so condition \eqref{eq:K*} holds with $K^*=[\beta A(1-\alpha)]^\frac{1}{1-\alpha}$. Then the left-hand side of \eqref{eq:necessity_diamond} becomes
\begin{equation*}
    F_K(K^*,1)=\frac{\alpha}{\beta(1-\alpha)}+1-\delta.
\end{equation*}
\end{exmp}

\subsection{Bewley model with idiosyncratic investment shocks}\label{subsec:bewley1}

We next consider an endogenous growth model with infinitely-lived agents subject to idiosyncratic investment shocks. For analytical tractability, we employ logarithmic preferences with a linear production function as in \citet{Kiyotaki1998}.\footnote{The model is also related to \citet{Kocherlakota2009}, who introduces an intrinsically useless asset and a collateral constraint in a model with idiosyncratic investment risk and labor. We modify the models of \citet{Kiyotaki1998} and \citet{Kocherlakota2009} by abstracting from labor and the credit market and introducing a dividend-paying asset. \citet*{HiranoJinnaiTodaLeverage} study an extension with a more general production function, endogenous dividends, and a credit market subject to a leverage constraint.}

\paragraph{Model}

There is a homogeneous good that can be consumed or used as capital. There is an asset in unit supply, which pays exogenous dividend $D_t\ge 0$ at time $t$ and cannot be shorted. There is a unit mass of a continuum of agents indexed by $i\in I=[0,1]$.\footnote{See \citet{SunZhang2009} for a mathematical foundation of models with a continuum of agents.} A typical agent has the logarithmic utility function
\begin{equation}
    \E_0\sum_{t=0}^\infty \beta^t\log c_{it}, \label{eq:logutility}
\end{equation}
where $\beta\in (0,1)$ is the discount factor and $c_{it}\ge 0$ is consumption.

Let $z_{it}$ be the productivity of agent $i$ at time $t$, which evolves over time according to finite-state Markov chain, independently across agents. The set of possible productivities is denoted by $\set{z_0,\dots,z_N}$ with $0=z_0<z_1<\dots<z_N$, so type 0 can be interpreted as savers (who have no entrepreneurial skill) and type $n\ge 1$ can be interpreted as investors, with higher types being more productive. Agents produce the good using the linear technology $y_{i,t+1}=z_{it}k_{it}$, where $k_{it}\ge 0$ is capital input at time $t$ and $y_{i,t+1}$ is output at time $t+1$.

\paragraph{Equilibrium}

The economy starts at $t=0$ with an initial specification of productivity and endowments of capital and asset $\set{(z_{i0},k_{i,-1},x_{i,-1})}_{i\in I}$, with $\int_I x_{i,-1}\diff i=1$. A \emph{rational expectations equilibrium} consists of sequences of asset prices $\set{P_t}_{t=0}^\infty$ and allocations $\set{(c_{it},k_{it},x_{it})_{i\in I}}_{t=0}^\infty$ such that
\begin{enumerate}
	\item each agent maximizes utility \eqref{eq:logutility} subject to the budget constraint
    \begin{equation}
        c_{it}+k_{it}+P_tx_{it}=w_t\eqqcolon z_{i,t-1}k_{i,t-1}+(P_t+D_t)x_{i,t-1} \label{eq:budget_bewley1}
    \end{equation}
    and the nonnegativity constraints $c_{it},k_{it},x_{it}\ge 0$, and
	\item the asset market clears, so $\int_I x_{it}\diff i=1$.
\end{enumerate}

We introduce the following technical condition.

\begin{asmp}\label{asmp:irreducible}
Let $\pi_{nn'}=\Pr(z_{i,t+1}=z_{n'} \mid z_{it}=z_n)$ be the transition probability. The transition probability matrix $\Pi=(\pi_{nn'})$ as well as its $N\times N$ submatrix $\Pi_1\coloneqq (\pi_{nn'})_{n,n'=1}^N$ are irreducible.
\end{asmp}

The irreducibility of $\Pi$ implies that the agent types are not permanent. The irreducibility of $\Pi_1$ implies that once agents become investors, with positive probability, they could visit all productivity states before returning to savers. Because low productivity agents are natural buyers of the asset, this assumption generates a demand for the asset, which allows us to bound the asset price from below. In what follows, for a square matrix $A$, let $\rho(A)$ denote its spectral radius (largest absolute value of all eigenvalues). The following theorem shows the necessity of asset price bubbles under the condition $R<G_d<G$.

\begin{thm}[Necessity of bubbles in Bewley model with investment shocks]\label{thm:bewley1}
Define the square nonnegative matrix $A\coloneqq (\beta z_n\pi_{nn'})$. If Assumption \ref{asmp:irreducible} holds and
\begin{equation}
    0<G_d\coloneqq \limsup_{t\to\infty} D_t^{1/t}<\rho(A)\eqqcolon G, \label{eq:necessity_bewley1}
\end{equation}
then all equilibria are asymptotically bubbly.
\end{thm}

Some remarks are in order. Since by assumption $z_0=0$, type 0 agents have no entrepreneurial skill and can save only through asset purchase. Under this condition, the counterfactual autarky interest rate is $R=0$. The numbers $G=\rho(A)$ and $G_d$ can be interpreted as (a lower bound of) the long-run economic and dividend growth rates. (The actual growth rate of the economy is of course endogenously determined because the risk-free rate is endogenous.) Thus the bubble necessity condition \eqref{eq:necessity_bewley1} exactly parallels that of Theorem \ref{thm:necessity} in \eqref{eq:necessity_OLG}.

Using $G=\rho(A)$, we can derive some comparative statics. Note that by Theorem 8.1.18 of \citet{HornJohnson2013}, $G=\rho(A)$ is increasing in each $z_n$. Therefore if the productivity of capital gets sufficiently high, the condition $G>G_d$ will be satisfied. This implies that technological innovations that enhance the overall productivity will inevitably generate asset price bubbles, which is consistent with the view in \citet[p.~22]{Scheinkman2014} that highlights the importance of the relationship between technological progress and asset price bubbles. Furthermore, applying an intermediate step of the proof of Proposition 5(iv) in \citet{BeareToda2022ECMA}, $G=\rho(A)$ is also increasing in the persistence of the Markov chain,\footnote{By ``increasing persistence'', we mean that we parameterize the transition probability matrix as $\tau I+(1-\tau)\Pi$ for $\tau\in [0,1)$ and increase $\tau$.} so if the persistence of productivity gets sufficiently high, bubbles inevitably arise.

\subsection{Bewley model with idiosyncratic preference shocks}\label{subsec:bewley2}

Finally, we consider a Bewley model with idiosyncratic preference shocks and endogenous labor supply. For analytical tractability, we employ quasi-linear preferences commonly used in monetary theory.\footnote{Representative papers include \citet{LagosWright2005} and \citet{RocheteauWright2005}, who study the welfare cost of inflation under various market structures (\eg, bargaining, price-taking, and price-posting). Our model is directly related to \citet{ChienWen2022}, who study optimal taxation in a heterogeneous-agent model with capital and labor. We modify this model by abstracting from capital and introducing a dividend-paying asset. The model of \citet{ChienWen2022} builds on \citet{Wen2015}, who studies a competitive equilibrium model with preference shocks and endogenous labor supply in a setting similar to \citet{LagosWright2005}.}

\paragraph{Model}

A representative firm produces the consumption good using the linear technology $Y_t=A_tL_t$, where $A_t>0$ is labor productivity and $L_t\ge 0$ is labor input. Thus the wage rate equals $A_t$. There is an asset in unit supply, which pays dividend $D_t\ge 0$ at time $t$. The sequence $\set{(A_t,D_t)}_{t=0}^\infty$ is exogenous and deterministic.

There is a unit mass of a continuum of agents indexed by $i\in I\in [0,1]$. A typical agent has the quasi-linear utility function
\begin{equation}
    \E_0\sum_{t=0}^\infty\beta^t[\theta_{it} u(c_{it})-\ell_{it}], \label{eq:ql}
\end{equation}
where $\beta\in (0,1)$ is the discount factor, $\theta_{it}>0$ is a preference shock, $c_{it}\ge 0$ is consumption, $\ell_{it}$ is labor supply, and $u$ is the period utility function. We assume that $u$ exhibits constant relative risk aversion $\gamma>0$ as in \eqref{eq:CRRA}. As is common in quasi-linear models, $\ell$ could be positive or negative, and we interpret the case $\ell<0$ as leisure $-\ell>0$.

The timing convention is as follows. At the beginning of period $t$, the agent first chooses the labor supply $\ell_{it}$. After choosing labor, the preference shock $\theta_{it}$ realizes, which is an \iid draw from a cumulative distribution function $F$ supported on $\Theta\coloneqq [\theta_L,\theta_H]$ with $0<\theta_L<\theta_H$. After observing $\theta_{it}$, the agent chooses consumption $c_{it}$ and asset holdings $x_{it}$. Intuitively, agents with high $\theta$ have an urge to consume because the weight on the utility from consumption is higher. Hence agents with low $\theta$ are the natural buyers of the asset.

\paragraph{Equilibrium}

The economy starts at $t=0$ with an initial endowment of asset $(x_{i,-1})_{i\in I}$, where $\int_I x_{i,-1}\diff i=1$. A \emph{rational expectations equilibrium} consists of sequences of wages and asset prices $\set{(A_t,P_t)}_{t=0}^\infty$ and allocations $\set{(c_{it},\ell_{it},x_{it})_{i\in I}}_{t=0}^\infty$ such that
\begin{enumerate}
    \item each agent maximizes utility \eqref{eq:ql} subject to the budget constraint
    \begin{equation}
        c_{it}+P_tx_{it}=w_{it}\coloneqq A_t\ell_{it}+(P_t+D_t)x_{i,t-1} \label{eq:budget_bewley2}
    \end{equation}
    and the shortsales constraint $x_{it}\ge 0$,
    \item the commodity market clears, so
    \begin{equation}
        \int_I c_{it}\diff i=A_t\int_I \ell_{it}\diff i+D_t, \label{eq:cclear}
    \end{equation}
    \item the asset market clears, so $\int_I x_{it}\diff i=1$.
\end{enumerate}

The following proposition characterizes the equilibrium dynamics.

\begin{prop}\label{prop:bewley2}
Letting $R_t\coloneqq (P_{t+1}+D_{t+1})/P_t$ be the equilibrium gross risk-free rate, there exists a threshold $\bar{\theta}_t\in \Theta$ such that the optimal consumption rule is
\begin{equation}
    c_t(\theta)=\left(\frac{A_{t+1}}{\beta R_t}\min\set{\theta,\bar{\theta}_t}\right)^{1/\gamma} \label{eq:crule}
\end{equation}
and the aggregate dynamics is
\begin{subequations}\label{eq:dynamics}
    \begin{align}
        \frac{1}{A_t}&=\frac{\beta R_t}{A_{t+1}}\int_\Theta\max\set{1,\theta/\bar{\theta}_t}\diff F(\theta), \label{eq:dynamics_A} \\
        P_t&=\left(\frac{A_{t+1}}{\beta R_t}\right)^{1/\gamma}\int_\Theta\max\set{0,\bar{\theta}_t^{1/\gamma}-\theta^{1/\gamma}}\diff F(\theta). \label{eq:dynamics_P} 
    \end{align}
\end{subequations}
\end{prop}

The intuition for Proposition \ref{prop:bewley2} is as follows. Due to quasi-linear preferences, agents adjust labor supply to achieve a common wealth $w_{it}=w_t$ in \eqref{eq:budget_bewley2}. Agents with high $\theta$ do not save and set $c_{it}=w_t$, while those with low $\theta$ choose savings $x_{it}$ to satisfy the Euler equation, which explains the cutoff rule \eqref{eq:crule}. Conditions \eqref{eq:dynamics_A} and \eqref{eq:dynamics_P} are essentially the (unconditional) Euler equation and the asset market clearing condition.

Monetary theory interprets the integral
\begin{equation}
    R(\bar{\theta}_t)\coloneqq \int_\Theta\max\set{1,\theta/\bar{\theta}_t}\diff F(\theta) \label{eq:liquidity_premium}
\end{equation}
as the liquidity premium \citep[Equation (12)]{Wen2015}.\footnote{See also \citet*{GeromichalosLicariSuarez-Lledo2007}, \citet{Lagos2010JME}, and \citet{RocheteauWright2013} for asset pricing implications of the liquidity premium in monetary search models.} This is because if agents are liquidity-unconstrained, then $\bar{\theta}_t=\theta_H$ and \eqref{eq:dynamics_A} reduces to $1/A_t=\beta R_t/A_{t+1}$, which is the usual Euler equation. In general, the liquidity premium arises because constrained agents are prevented from shortselling the asset, which raises its price. To see this formally, combining \eqref{eq:dynamics_A} and \eqref{eq:dynamics_P}, we may write
\begin{equation}
    P_t=A_t^{1/\gamma}\left(\int_\Theta \max\set{1,\theta/\bar{\theta}_t}\diff F(\theta)\right)^{1/\gamma}\int_\Theta \max\set{0,\bar{\theta}_t^{1/\gamma}-\theta^{1/\gamma}}\diff F(\theta), \label{eq:Ptheta}
\end{equation}
so $P_t$ is proportional to $R(\bar{\theta}_t)^{1/\gamma}$ using \eqref{eq:liquidity_premium}. However, the liquidity premium does not affect whether the asset price exhibits a bubble or not because $R(\bar{\theta}_t)\in [1,\theta_H/\theta_L]$ is bounded (see Lemma \ref{lem:bubble}).

We introduce the following technical condition.

\begin{asmp}\label{asmp:delta}
There exists $\delta>0$ such that $0<F(\theta_L)=F(\theta_L+\delta)<1$.
\end{asmp}

Assumption \ref{asmp:delta} states that $\theta=\theta_L$ is an isolated point mass of the cumulative distribution function $F$. Assumption \ref{asmp:delta} holds, for example, if $F$ is a nondegenerate distribution taking finitely many values. Because agents with low $\theta$ are natural buyers of the asset, this assumption generates a demand for the asset, which allows us to bound the asset price from below using \eqref{eq:Ptheta}. The following theorem shows the necessity of asset price bubbles under the condition $R<G_d<G$.

\begin{thm}[Necessity of bubbles in Bewley model with preference shocks]\label{thm:bewley2}
If Assumption \ref{asmp:delta} holds and
\begin{equation}
    0<G_d\coloneqq \limsup_{t\to\infty} D_t^{1/t}<\liminf_{t\to\infty}A_t^{1/\gamma t}\eqqcolon G, \label{eq:necessity_bewley2}
\end{equation}
then all equilibria are asymptotically bubbly.
\end{thm}

The interpretation of the bubble necessity condition \eqref{eq:necessity_bewley2} is similar to that of Theorem \ref{thm:bewley1}. To see why, consider a complete-market setting where the preference shock $\theta_{it}=\theta$ is constant. Then setting $\ell_{it}=\ell_t$ and $x_{it}=1$ in the budget constraint \eqref{eq:budget_bewley2}, individual consumption is $c_{it}=c_t=A_t\ell_t+D_t$. Maximizing the utility function \eqref{eq:ql} (with relative risk aversion $\gamma$) with respect to $\ell_t$ yields the first-order condition
\begin{equation*}
    \theta c_t^{-\gamma}A_t=1\iff c_t=(\theta A_t)^{1/\gamma}.
\end{equation*}
Thus a lower bound of long-run economic (consumption) growth is
\begin{equation*}
    \liminf_{t\to\infty}c_t^{1/t}=\liminf_{t\to\infty} A_t^{1/\gamma t},
\end{equation*}
which explains the condition \eqref{eq:necessity_bewley2}.

\section{Conclusion}

In this paper we presented a conceptually new perspective on thinking about asset price bubbles: their \emph{necessity}. We showed a plausible general class of economic models with dividend-paying assets in which the emergence of bubbles is a necessity by proving that all equilibria are asymptotically bubbly. This surprising insight of the necessity of bubbles is fundamentally different from economists' long-held view that bubbles are either not possible in rational equilibrium models or even if they are, a situation in which bubbles occur is a special circumstance and hence fragile. This is also conceptually different from rational bubble models (most of which are monetary models) that show the possibility of bubbles. We emphasize that the necessity of bubbles naturally arises in workhorse models in macro-finance. Hence, our Bubble Necessity Theorem challenges the conventional wisdom on bubbles and may open up a new direction for research.

In a series of working papers, we and our collaborators apply the idea of the necessity of asset price bubbles to macro-finance \citep*{HiranoJinnaiTodaLeverage}, housing \citep{HiranoTodaHousingbubble}, and growth \citep{HiranoTodaUnbalanced}. Interestingly and importantly, in all of these models, bubbles naturally and necessarily arise.

\appendix

\section{Proofs}\label{sec:proof}

\subsection{Proof of Theorem \ref{thm:exist}}

We need several lemmas to prove Theorem \ref{thm:exist}.

\begin{lem}\label{lem:asset_pricing}
In equilibrium, the asset pricing equation \eqref{eq:foc} holds.
\end{lem}

\begin{proof}
Take any equilibrium. To simplify notation, let $U_t=U$, $P_t=P$, $P_{t+1}=P'$, and $D_{t+1}=D'$, etc. Using the budget constraint \eqref{eq:budget} to eliminate $y,z'$, the young seek to solve
\begin{subequations}
\begin{align}
    &\maximize && U(a-Px,b'+(P'+D')x) \label{eq:Ux}\\
    &\st && a-Px\ge 0, \label{eq:y>0}\\
    &   && b'+(P'+D')x\ge 0. \label{eq:z>0}
\end{align}
\end{subequations}
In equilibrium, market clearing forces $x=1$. Since $b'>0$ and $P'+D'\ge 0$, the nonnegativity constraint \eqref{eq:z>0} never binds at $x=1$. Let $\lambda\ge 0$ be the Lagrange multiplier associated with the nonnegativity constraint \eqref{eq:y>0} and let
\begin{equation*}
    L(x,\lambda)=U(a-Px,b'+(P'+D')x)+\lambda(a-Px)
\end{equation*}
be the Lagrangian. The first-order condition implies
\begin{equation}
    -PU_y+(P'+D')U_z-\lambda P=0\iff P=M(P'+D')-\frac{\lambda}{U_y}P, \label{eq:foc2}
\end{equation}
where $M=U_z/U_y$ and $U_y,U_z$ are evaluated at $(a-P,b'+P'+D')$. If $P<a$, then the nonnegativity constraint \eqref{eq:y>0} does not bind, $\lambda=0$, and \eqref{eq:foc2} reduces to $P=M(P'+D')$ and \eqref{eq:foc} holds. If $P=a$, then the nonnegativity constraint \eqref{eq:y>0} binds, $\lambda\ge 0$, and \eqref{eq:foc2} implies $M(P'+D')=a+(\lambda/U_y)a\ge a$, so \eqref{eq:foc} holds.
\end{proof}

\begin{lem}\label{lem:backward_induction}
For all $P_{t+1}\ge 0$, there exists $P_t\in [0,a_t]$ that satisfies \eqref{eq:foc}.
\end{lem}

Define the function $f:[0,a]\to \R$ by
\begin{equation*}
    f(P)=(P'+D')M(a-P,b'+P'+D')-P.
\end{equation*}
By Assumption \ref{asmp:U}, $f$ is continuous. Since $U$ is quasi-concave, the marginal rate of substitution $M$ in \eqref{eq:MRS} is increasing in $y$. Therefore $M(a-P,b'+P'+D')$ is decreasing in $P$, so $f$ is strictly decreasing. Clearly
\begin{equation*}
    f(0)=(P'+D')M(a,b'+P'+D')\ge 0.
\end{equation*}
If $f(a)>0$, then the definition of $f$ implies that
\begin{equation*}
    a<(P'+D')M(0,b'+P'+D'),
\end{equation*}
so \eqref{eq:foc} holds with $P=a$. If $f(a)\le 0$, by the intermediate value theorem, there exists $P\in [0,a]$ such that $f(P)=0$, which clearly satisfies \eqref{eq:foc}.

\begin{proof}[Proof of Theorem \ref{thm:exist}]
Although the existence of equilibrium follows from \citet[Theorem 1]{Wilson1981}, to make the paper self-contained, we present a standard truncation argument as in \citet{BalaskoShell1980}. Define the set $\cA\coloneqq \prod_{t=0}^\infty [0,a_t]$ endowed with the product topology induced by the Euclidean topology on $[0,a_t]\subset \R$ for all $t$. By Tychonoff's theorem, $\cA$ is nonempty and compact.

Define a \emph{$T$-equilibrium} by a sequence $\set{P_t}_{t=0}^\infty$ such that $P_t\in [0,a_t]$ for all $t$ and the asset pricing equation \eqref{eq:foc} holds for $t=0,\dots,T-1$. Let $\cP_T\subset \cA$ be the set of all $T$-equilibria. For any sequence $\set{P_t}_{t=T}^\infty$ such that $P_t\in [0,a_t]$ for all $t\ge T$, repeatedly applying Lemma \ref{lem:backward_induction}, by backward induction we can construct a $T$-equilibrium $\set{P_t}_{t=0}^\infty$. Therefore $\cP_T\neq\emptyset$.

Since $U_t$ is continuously differentiable, the marginal rate of substitution $M_t$ is continuous, and hence $\cP_T$ is closed. Furthermore, by the definition of the $T$-equilibrium, we have $\cP_T\supset \cP_{T+1}$ for all $T$. Since $\cP_T\subset \cA$ and $\cA$ is compact, we have $\cP\coloneqq \bigcap_{t=0}^\infty \cP_t\neq\emptyset$. If we take any $\set{P_t}_{t=0}^\infty\in \cP$, by definition \eqref{eq:foc} holds for all $t$. The quasi-concavity of $U_t$ implies that we have an equilibrium.
\end{proof}

\subsection{Proof of Theorem \ref{thm:necessity}}

Take any equilibrium $\set{P_t}_{t=0}^\infty$. Let $p_t=P_t/a_t$ and $d_t=D_t/a_t$ be the asset price and dividend detrended by the endowment of the young. Let $G_t=a_{t+1}/a_t$ be the endowment growth rate and $w_t=b_t/a_t$ be the old-to-young endowment ratio. Since $G_d>0$, we have $D_t>0$ infinitely often, so by the remark after Lemma \ref{lem:bubble}, we have $0<P_t\le a_t$. Dividing both sides by $a_t>0$, we obtain $0<p_t\le 1$. Since $D_t\ge 0$, we have $d_t=D_t/a_t\ge 0$.

We need several lemmas to prove Theorem \ref{thm:necessity}.

\begin{lem}\label{lem:Da}
We have $\sum_{t=1}^\infty D_t/a_t<\infty$. In particular, $\lim_{t\to\infty}d_t=0$.
\end{lem}

\begin{proof}
Using \eqref{eq:G} and \eqref{eq:necessity_OLG}, we can take $\epsilon>0$ and $T>0$ such that
\begin{equation*}
    D_t^{1/t}<G_d+\epsilon<G-\epsilon<a_{t+1}/a_t
\end{equation*}
for $t\ge T$. Therefore
\begin{equation*}
    \frac{D_t}{a_t}=\frac{D_t}{a_T}\left(\prod_{s=T}^{t-1} \frac{a_{s+1}}{a_s}\right)^{-1}<\frac{(G_d+\epsilon)^t}{a_T}(G-\epsilon)^{T-t}=\frac{(G-\epsilon)^T}{a_T}\left(\frac{G_d+\epsilon}{G-\epsilon}\right)^t,
\end{equation*}
which is summable because $G_d+\epsilon<G-\epsilon$.
\end{proof}

\begin{lem}\label{lem:p_ratio}
The following statement is true:
\begin{equation*}
    (\exists r>0)(\exists T>0)(\forall t\ge T)\quad p_t\in (0,1/2)\implies \frac{p_{t+1}}{p_t}\le r.
\end{equation*}
In other words, there exists a universal constant $r>0$ such that $p_{t+1}/p_t\le r$ for all large enough $t$ whenever $p_t<1/2$.
\end{lem}

\begin{proof}
Suppose $p_t<1/2$. Then in particular $p_t<1$ and $P_t<a_t$, so \eqref{eq:foc} holds without the min operator with $a_t$. Dividing both sides by $a_t>0$ and using the definition of the forward rate function in \eqref{eq:ft}, we obtain
\begin{equation}
    \frac{G_{t+1}(p_{t+1}+d_{t+1})}{p_t}=f_t(1-p_t,G_{t+1}(w_{t+1}+p_{t+1}+d_{t+1})). \label{eq:foc/a}
\end{equation}
With a slight abuse of notation, let
\begin{equation}
    (y_t,z_t)\coloneqq (1-p_t,G_{t+1}(w_{t+1}+p_{t+1}+d_{t+1})). \label{eq:yz}
\end{equation}
Then using $d_t\ge 0$ and \eqref{eq:foc/a}, we obtain 
\begin{equation}
    0<\frac{p_{t+1}}{p_t}\le \frac{p_{t+1}+d_{t+1}}{p_t}=\frac{1}{G_{t+1}}f_t(y_t,z_t). \label{eq:p_ratio}
\end{equation}

By Assumption \ref{asmp:ab}, we can take $0<\ubar{G}<G<\bar{G}$ and $0\le \ubar{w}\le w<\bar{w}$ such that $G_t\in (\ubar{G},\bar{G})$ and $w_t\in [\ubar{w},\bar{w})$ for large enough $t$. Define the compact set
\begin{equation}
    K\coloneqq [1/2,1]\times [\ubar{G}\ubar{w},\bar{G}(\bar{w}+1)] \subset \R_{++}\times \R_+. \label{eq:K}
\end{equation}
Since $0<p_{t+1}\le 1$ and $d_{t+1}\to 0$ by Lemma \ref{lem:Da}, it follows from the definitions of $(y_t,z_t)$ in \eqref{eq:yz} and $K$ in \eqref{eq:K} that
\begin{equation}
    (\exists T_1)(\forall t\ge T_1)\quad p_t<1/2 \implies (y_t,z_t)\in K, \label{eq:yzinK}
\end{equation}
that is, $(y_t,z_t)\in K$ for all large enough $t$ whenever $p_t<1/2$. In general, for any nonempty compact set $K\subset \R_{++}\times \R_+$, define
\begin{equation}
    0\le \bar{f}(K)\coloneqq \max_{(y,z)\in K}f(y,z), \label{eq:fK}
\end{equation}
which is well defined because $f$ is continuous by Assumption \ref{asmp:f}. By Assumption \ref{asmp:f}, we have
\begin{equation*}
    (\exists T_2>0)(\forall t\ge T_2)\quad (y,z)\in K \implies \abs{f_t(y,z)-f(y,z)}<1.
\end{equation*}
In particular, if $t\ge T_2$, by the definition of $\bar{f}$ is \eqref{eq:fK}, we have
\begin{equation}
    (\forall t\ge T_2)\quad (y_t,z_t)\in K \implies f_t(y_t,z_t)\le \bar{f}(K)+1. \label{eq:ft_ub}
\end{equation}

Define $T=\max\set{T_1,T_2}$. If $t\ge T$ and $p_t<1/2$, then \eqref{eq:yzinK} implies $(y_t,z_t)\in K$. Therefore putting all the pieces together, we obtain
\begin{align*}
    \frac{p_{t+1}}{p_t}&\le \frac{1}{G_{t+1}}f_t(y_t,z_t) && (\because \eqref{eq:p_ratio})\\
    &\le \frac{1}{\ubar{G}}(\bar{f}(K)+1)\eqqcolon r. && (\because G_{t+1}\ge \ubar{G}, \eqref{eq:ft_ub}) \qedhere
\end{align*}
\end{proof}

\begin{lem}\label{lem:p_ratio2}
The following statement is true:
\begin{equation}
    (\exists \epsilon>0)(\exists T>0)(\forall t\ge T) \quad p_t<\epsilon\implies \frac{p_{t+1}}{p_t}\le \frac{G_d}{G}<1. \label{eq:statement2}
\end{equation}
\end{lem}

\begin{proof}
Take $r>0$ and $T>0$ as in Lemma \ref{lem:p_ratio}. For $\epsilon \in (0,1/2)$, define the compact set
\begin{equation}
    K(\epsilon)\coloneqq [1-\epsilon,1]\times [\ubar{G}\ubar{w},\bar{G}(\bar{w}+r\epsilon)]\subset \R_{++}\times \R_+. \label{eq:Kepsilon}
\end{equation}
If $t\ge T$ and $p_t<\epsilon$, then in particular $p_t<1/2$. Therefore by Lemma \ref{lem:p_ratio}, we have
\begin{equation*}
    \frac{p_{t+1}}{p_t}\le r\implies p_{t+1}\le rp_t\le r\epsilon.
\end{equation*}
Therefore by the definition of $K(\epsilon)$ in \eqref{eq:Kepsilon}, we have $(y_t,z_t)\in K(\epsilon)$. For any $\delta>0$, by Assumption \ref{asmp:f} we have
\begin{equation}
    \sup_{(y,z)\in K(\epsilon)}\abs{f_t(y,z)-f(y,z)}\le \delta \label{eq:ft_bound}
\end{equation}
for $t\ge T$ (by choosing a larger $T$ if necessary). Since $(y_t,z_t)\in K(\epsilon)$, it follows from \eqref{eq:p_ratio} that
\begin{align}
    \frac{p_{t+1}}{p_t}&\le \frac{1}{G_{t+1}}f_t(y_t,z_t) \le \frac{1}{\ubar{G}}(f(y_t,z_t)+\delta) && (\because \eqref{eq:p_ratio}, \eqref{eq:ft_bound}) \notag \\
    &\le \frac{1}{\ubar{G}}(f(1-\epsilon,\bar{G}(\bar{w}+r\epsilon))+\delta), \label{eq:p_ratio2}
\end{align}
where the last line follows from the quasi-concavity of $U_t$ (hence $f(y,z)$ is decreasing in $y$ and increasing in $z$) and the fact that $(y_t,z_t)\in K(\epsilon)$ with $K(\epsilon)$ defined as in \eqref{eq:Kepsilon}.

By Assumption \ref{asmp:ab}, we may take $\ubar{G},\bar{G}$ arbitrarily close to $G$ and $\bar{w}$ arbitrarily close to $w$. Clearly, we can take $\epsilon, \delta>0$ arbitrarily close to zero. Therefore the right-hand side of \eqref{eq:p_ratio2} can be made arbitrarily close to $f(1,Gw)/G$, which is less than $G_d/G<1$ by condition \eqref{eq:necessity_OLG}. Therefore \eqref{eq:statement2} holds.
\end{proof}

\begin{lem}\label{lem:relevant}
In all equilibria, the asset is asymptotically relevant.
\end{lem}

\begin{proof}
Choose $\epsilon,T$ as in Lemma \ref{lem:p_ratio2}. By way of contradiction, suppose there exists an equilibrium in which the asset is asymptotically irrelevant. Then by Definition \ref{defn:asym_bubble} we can take $t_0\ge T$ such that $p_{t_0}<\epsilon$. Let us show by induction that $p_t<(G_d/G)^{t-t_{0}}\epsilon$ for all $t\ge t_0$. The claim is trivial when $t=t_0$. If the claim holds for some $t$, then in particular $p_t<\epsilon$, so using \eqref{eq:statement2} we obtain
\begin{equation*}
    p_{t+1}\le (G_d/G)p_t<(G_d/G)(G_d/G)^{t-t_0}\epsilon=(G_d/G)^{t+1-t_0}\epsilon,
\end{equation*}
so the claim holds for $t+1$ as well.

Since $G_d/G<1$, we have $p_t\to 0$. Multiplying both sides of \eqref{eq:p_ratio} by $G_{t+1}$ and letting $t\to\infty$, by Assumptions \ref{asmp:ab} and \ref{asmp:f}, we obtain
\begin{equation}
    R_t=G_{t+1}\frac{p_{t+1}+d_{t+1}}{p_t}=f_t(y_t,z_t)\to f(1,Gw). \label{eq:Rt_lim_gen}
\end{equation}
By \eqref{eq:necessity_OLG} and \eqref{eq:Rt_lim_gen}, we can take $\epsilon>0$ and $T>0$ such that
\begin{equation*}
    R_t<f(1,Gw)+\epsilon<G_d-\epsilon
\end{equation*}
for $t\ge T$. Furthermore, by \eqref{eq:necessity_OLG}, we have $D_t^{1/t}>G_d-\epsilon$ infinitely often. Therefore for such $t$, the present value of $D_t$ can be bounded from below as
\begin{equation*}
    q_tD_t=q_T(q_t/q_T)D_t\ge q_T(G_d-\epsilon)^{T-t}(G_d-\epsilon)^t=q_T(G_d-\epsilon)^T.
\end{equation*}
Since the lower bound is positive and does not depend on $t$, and there are infinitely many such $t$, we obtain $P_0\ge V_0=\sum_{t=1}^\infty q_tD_t=\infty$, which is a contradiction.
\end{proof}

\begin{proof}[Proof of Theorem \ref{thm:necessity}]
Take any equilibrium. By Lemma \ref{lem:relevant}, the asset is asymptotically relevant. By Definition \ref{defn:asym_bubble}, we can take $T>0$ and $p>0$ such that $P_t/a_t\ge p$ for $t\ge T$. Then
\begin{equation*}
    \sum_{t=1}^\infty \frac{D_t}{P_t}\le \sum_{t=1}^{T-1} \frac{D_t}{P_t}+\sum_{t=T}^\infty \frac{D_t}{pa_t}<\infty
\end{equation*}
by Lemma \ref{lem:Da}. Lemma \ref{lem:bubble} implies that the equilibrium is bubbly, and it is asymptotically bubbly because the asset is asymptotically relevant.
\end{proof}

\subsection{Proof of Theorem \ref{thm:diamond}}

Since $F$ is homogeneous of degree 1 and concave, $F_L(K,1)=F_L(1,1/K)$ is increasing in $K$. We need several lemmas to prove Theorem \ref{thm:diamond}.

\begin{lem}\label{lem:global}
If Assumption \ref{asmp:FL} holds, the map $(0,\infty)\ni K\mapsto \beta F_L(K,1) \in (0,\infty)$ has a unique fixed point $K^*>0$, which is globally stable. 
\end{lem}

\begin{proof}
By \eqref{eq:K*}, clearly $K^*$ is the unique fixed point of $K\mapsto \beta F_L(K,1)$. Take any sequence $\set{K_t}\subset (0,\infty)$ such that $K_{t+1}=\beta F_L(K_t,1)$. If $K_0\le K^*$, then \eqref{eq:K*} and the monotonicity of $F_L$ imply
\begin{equation*}
    K_0\le K_1=\beta F_L(K_0,1)\le \beta F_L(K^*,1)=K^*.
\end{equation*}
Continuing this argument, we have $K_0\le K_1\le \dots\le K_t\le K^*$. Since $\set{K_t}$ is bounded and monotonically increasing, it is convergent. The limit is a fixed point of $K\mapsto \beta F_L(K,1)$, and by \eqref{eq:K*}, it must be $K^*$. The same argument applies when $K_0\ge K^*$. Therefore $K^*>0$ is the unique and globally stable fixed point of $K\mapsto \beta F_L(K,1)$.
\end{proof}

\begin{lem}\label{lem:K_ub}
In any equilibrium, we have $\limsup_{t\to\infty}K_t\le K^*$.
\end{lem}

\begin{proof}
By \eqref{eq:budget_diamond}, we have $K_{t+1}\le \beta F_L(K_t,1)$. Define the sequence $\set{\bar{K}_t}$ by $\bar{K}_0=K_0$ and $\bar{K}_{t+1}=\beta F_L(\bar{K}_t,1)$. Since $F_L(K,1)$ is increasing in $K$, by induction we have $K_t\le \bar{K}_t$ for all $t$. Since by Lemma \ref{lem:global} the map $K\mapsto \beta F_L(K,1)$ is globally stable, it follows that $\limsup_{t\to\infty}K_t\le \lim_{t\to\infty}\bar{K}_t=K^*$.
\end{proof}

\begin{lem}\label{lem:PK}
Suppose there exists a subsequence such that $\lim_{n\to\infty}(P_{t_n},K_{t_n})=(0,k_0)$ for some $k_0>0$. Define $\set{k_j}\subset (0,\infty)$ recursively by
\begin{equation}
     k_{j+1}=\beta F_L(k_j,1). \label{eq:kj}
\end{equation}
Then $\lim_{n\to\infty}(P_{t_n+j},K_{t_n+j})=(0,k_j)$ for all $j$.
\end{lem}

\begin{proof}
We show the claim by induction on $j$. The claim holds for $j=0$ by assumption. Suppose the claim holds for some $j$. Letting $t=t_n+j$ and $n\to\infty$ in \eqref{eq:budget_diamond}, we obtain $K_{t_n+j+1}\to \beta F_L(k_j,1)\eqqcolon k_{j+1}$ by \eqref{eq:kj}. Rewriting \eqref{eq:noarbitrage_diamond} as
\begin{equation*}
    P_{t+1}=P_tF_K(K_{t+1},1)-D_{t+1}\le P_tF_K(K_{t+1},1)
\end{equation*}
and letting $t=t_n+j$ and $n\to\infty$, it follows that $P_{t_n+j+1}\to 0$.
\end{proof}

\begin{lem}\label{lem:K_bound}
For $\epsilon\in [0,K^*)$, define
\begin{subequations}
\begin{align}
    p(\epsilon)&\coloneqq \beta F_L(K^*-\epsilon,1)-(K^*-\epsilon)\ge 0, \label{eq:p_epsilon}\\
    r(\epsilon)&\coloneqq F_K(K^*-\epsilon,1)>0. \label{eq:r_epsilon}
\end{align}
\end{subequations}
If \eqref{eq:necessity_diamond} holds, $\liminf_{t\to\infty}P_t=0$, and $\liminf_{t\to\infty}K_t>0$, then there exist $\epsilon\in (0,K^*)$ with $r(\epsilon)<1$ and $T>0$ such that for all $t\ge T$, we have $K_t\in (K^*-\epsilon,K^*+\epsilon)$ and $P_t\le r(\epsilon)^{t-T}p(\epsilon)$.
\end{lem}

\begin{proof}
By \eqref{eq:necessity_diamond} and the definition of $r(\epsilon)$ in \eqref{eq:r_epsilon}, we can choose sufficiently small $\epsilon>0$ such that $r(\epsilon)<1$. Since $\liminf_{t\to\infty}P_t=0$ and $\liminf_{t\to\infty}K_t>0$, we can take a subsequence such that $(P_{t_n},K_{t_n})\to (0,k_0)$ for some $k_0>0$. Let $\set{k_j}\subset (0,\infty)$ be as in Lemma \ref{lem:PK}. Then by \eqref{eq:kj} and Lemma \ref{lem:global}, we have $k_j\to K^*$ as $j\to\infty$. Therefore we can take large enough $j$ such that $k_j\in (K^*-\epsilon,K^*+\epsilon)$. For this $j$, since by Lemma \ref{lem:PK} we have $P_{t_n+j}\to 0$ and $K_{t_n+j}\to k_j$, we can take large enough $n$ such that $K_{t_n+j}\in (K^*-\epsilon,K^*+\epsilon)$ and $P_{t_n+j}\le p(\epsilon)$. Define $T=t_n+j$. Since $n$ can be taken arbitrarily large, by Lemma \ref{lem:K_ub}, without loss of generality we may assume $K_t<K^*+\epsilon$ for all $t\ge T$.

Let us show by induction that $K_t\in (K^*-\epsilon,K^*+\epsilon)$ and $P_t\le r(\epsilon)^{t-T}p(\epsilon)$ for all $t\ge T$. The claim is obvious for $t=T$. Suppose the claim holds for some $t$. As mentioned before, $K_t<K^*+\epsilon$ holds. Regarding the lower bound, we obtain
\begin{align*}
	K_{t+1}&=\beta F_L(K_t,1)-P_t && (\because \eqref{eq:budget_diamond})\\
	&\ge \beta F_L(K^*-\epsilon,1)-r(\epsilon)^{t-T}p(\epsilon) && (\because \text{induction hypothesis})\\
	&\ge \beta F_L(K^*-\epsilon,1)-p(\epsilon) && (\because r(\epsilon)<1)\\
	&=K^*-\epsilon. && (\because \text{definition of $p(\epsilon)$ in \eqref{eq:p_epsilon}})
\end{align*}
Furthermore, using the concavity of $F$ (so $F_K(K,1)$ is decreasing),
\begin{align*}
	P_{t+1}&\le \frac{P_{t+1}+D_{t+1}}{P_t}P_t && (\because D_{t+1}\ge 0)\\
	&=P_tF_K(K_{t+1},1) && (\because \eqref{eq:noarbitrage_diamond})\\
	&\le P_tF_K(K^*-\epsilon) && (\because K_{t+1}\ge K^*-\epsilon)\\
	&=r(\epsilon)^{t+1-T}p(\epsilon). && (\because P_t\le r(\epsilon)^{t-T}p(\epsilon), \eqref{eq:r_epsilon})
\end{align*}
Therefore the claim holds for $t+1$ as well.
\end{proof}

\begin{proof}[Proof of Theorem \ref{thm:diamond}]
Suppose that there exists an equilibrium with $\liminf_{t\to\infty}K_t>0$ such that the asset is asymptotically irrelevant, so $\liminf_{t\to\infty}P_t=0$. By condition \eqref{eq:necessity_diamond} and \eqref{eq:r_epsilon}, we can take sufficiently small $\epsilon>0$ such that $r(\epsilon)<G_d<1$. Applying Lemma \ref{lem:K_bound}, we can take $T>0$ such that $K_t\in (K^*-\epsilon,K^*+\epsilon)$ for $t\ge T$. Using \eqref{eq:noarbitrage_diamond}, the gross risk-free rate for $t\ge T$ is
\begin{equation*}
    R_t=F_K(K_{t+1},1)\le F_K(K^*-\epsilon,1)=r(\epsilon)<G_d.
\end{equation*}
By the same argument as in the proof of Lemma \ref{lem:relevant}, we have $P_0\ge V_0=\infty$, which is a contradiction.

Therefore in all equilibria, the asset is asymptotically relevant, and there exists $p>0$ such that $P_t\ge p$. Since $D_t$ asymptotically grows at rate $G_d<1$, we have $\sum_{t=1}^\infty D_t/P_t<\infty$, so by Lemma \ref{lem:bubble} the equilibrium is asymptotically bubbly.
\end{proof}

\subsection{Proof of Theorem \ref{thm:bewley1}}

We need several lemmas to prove Theorem \ref{thm:bewley1}.

\begin{lem}\label{lem:Wt}
Let $W_{nt}$ be the aggregate wealth held by type $n$ agents at time $t$ and $v_t'=(W_{0t},\dots,W_{Nt})$ be the row vector of aggregate wealth. Then $v_t'\ge v_0'A^t$.
\end{lem}

\begin{proof}
Let $R_t=(P_{t+1}+D_{t+1})/P_t$ be the gross risk-free rate between time $t$ and $t+1$. Due to log utility, the optimal consumption rule is $c_{it}=(1-\beta)w_{it}$. Savings is thus $\beta w_{it}$. Because the productivity is predetermined, it is optimal for an agent to invest entirely in the technology (asset) if $z_{it}>R_t$ ($z_{it}<R_t$). If $z_{it}=R_t$, the agent is indifferent between the technology and asset. Therefore using the budget constraint \eqref{eq:budget_bewley1}, individual wealth evolves according to
\begin{equation}
    w_{i,t+1}=\beta \max\set{z_{it},R_t} w_{it}. \label{eq:wt}
\end{equation}

Let $W_{nt}$ be the aggregate wealth held by type $n$ agents. Then aggregating \eqref{eq:wt} across agents and applying the strong law of large numbers, we obtain
\begin{equation}
    W_{n',t+1}=\sum_{n=0}^N\pi_{nn'}\beta \max\set{z_n,R_t}W_{nt}\ge \sum_{n=0}^N\pi_{nn'}\beta z_nW_{nt}. \label{eq:Wt}
\end{equation}
Collecting the terms in \eqref{eq:Wt} into a row vector and using the definitions of $v_t'$ and $A$, we obtain $v_{t+1}'\ge v_t'A$. Iterating this inequality, we obtain $v_t'\ge v_0'A^t$.
\end{proof}

\begin{lem}\label{lem:W0}
There exists a constant $w_0>0$ such that $W_{0t}\ge w_0G^t$ for all $t$.
\end{lem}

\begin{proof}
Noting that $z_0=0$, partition the matrix $A=(\beta z_n\pi_{nn'})$ as
\begin{equation}
    A=\begin{bmatrix}
        0 & 0\\
        b_1 & A_1
    \end{bmatrix}, \label{eq:A_partition}
\end{equation}
where $A_1=(\beta z_n\pi_{nn'})_{n,n'=1}^N$ is the $N\times N$ submatrix. Since $\Pi$ is irreducible and $z_n>0$ for all $n\ge 1$, we have $b_1>0$. Similarly, since $\Pi_1=(\pi_{nn'})_{n,n'=1}^N$ is irreducible by Assumption \ref{asmp:irreducible} and $z_n>0$ for all $n\ge 1$, $A_1$ is irreducible.

Since $A$ is nonnegative, by Theorem 8.3.1 of \citet{HornJohnson2013}, the spectral radius $\rho(A)$ is an eigenvalue with a corresponding nonnegative left eigenvector $u'$. Partition $u'$ as $u'=(u_0,u_1')$. Multiplying $u'$ to $A$ in \eqref{eq:A_partition} from the left and comparing entries, we obtain
\begin{subequations}
\begin{align}
    \rho(A)u_0&=u_1'b_1, \label{eq:u0}\\
    \rho(A)u_1'&=u_1'A_1. \label{eq:u1}
\end{align}
\end{subequations}
If $u_1=0$, then \eqref{eq:u0} and $\rho(A)=G>0$ implies $u_0=0$. Then $u=0$, which contradicts the fact that $u$ is an eigenvector of $A$. Therefore $u_1\neq 0$, and \eqref{eq:u1} implies that $u_1'$ is a nonnegative left eigenvector of $A_1$. Since $A$ is block lower triangular, we have $0<\rho(A)=\rho(A_1)$. Since $A_1$ is irreducible, by the Perron-Frobenius theorem \citep[Theorem 8.4.4]{HornJohnson2013}, $u_1'$ must be the left Perron vector of $A_1$ and hence $u_1\gg 0$. Therefore \eqref{eq:u0} implies $u_0=u_1'b_1/\rho(A)>0$, so $u'=(u_0,u_1')\gg 0$.

Since agents are endowed with positive endowments, the vector of initial aggregate wealth $v_0'$ is positive. Therefore we can take $\epsilon>0$ such that $v_0'\ge \epsilon u'$. By Lemma \ref{lem:Wt}, we obtain
\begin{equation*}
    v_t'\ge v_0'A^t\ge \epsilon u'A^t=\epsilon \rho(A)^tu'=\epsilon G^tu'.
\end{equation*}
Comparing the 0-th entry, we obtain $W_{0t}\ge \epsilon u_0 G^t$, so we can take $w_0\coloneqq \epsilon u_0$.
\end{proof}

\begin{proof}[Proof of Theorem \ref{thm:bewley1}]
Using the definition of $G_d$ and the assumption $G_d<G$, it follows from Lemma \ref{lem:W0} that we can take $\epsilon>0$ such that $G>G_d+\epsilon$ and
\begin{equation}
    W_{0t}\ge w_0G^t>(G_d+\epsilon)^t\ge D_t \label{eq:WD}
\end{equation}
for large enough $t$. Since $z_0=0$, type 0 agents invest all wealth in the asset, so the market capitalization of the asset (which equals the the asset price because it is in unit supply) must exceed the aggregate savings of type 0: $P_t\ge \beta W_{0t}$. Therefore using \eqref{eq:WD}, for large enough $t$ we can bound the dividend yield from above as
\begin{equation*}
    \frac{D_t}{P_t}\le \frac{D_t}{\beta W_{0t}}\le \frac{1}{\beta w_0}\left(\frac{G_d+\epsilon}{G}\right)^t,
\end{equation*}
which is summable. By Lemma \ref{lem:bubble} the equilibrium is bubbly. Furthermore, $P_t\ge \beta W_{0t}\ge \beta w_0G^t$ implies that the equilibrium is asymptotically bubbly.
\end{proof}

\subsection{Proof of Proposition \ref{prop:bewley2} and Theorem \ref{thm:bewley2}}

We need several lemmas to prove Proposition \ref{prop:bewley2}. The following lemma is a consequence of quasi-linear utility.

\begin{lem}\label{lem:wt}
In equilibrium, there exists a sequence $\set{w_t}_{t=0}^\infty$ such that $w_{it}=w_t$ for all $i$: agents adjust labor supply $\ell_{it}$ to achieve a common wealth $w_{it}$ in \eqref{eq:budget_bewley2}.
\end{lem}

\begin{proof}
Using the budget constraint \eqref{eq:budget_bewley2} to eliminate consumption and labor, the continuation utility at time $t$ is
\begin{equation*}
    \E_\theta\left[\theta u(w_{it}-P_tx_{it})-\frac{w_{it}-(P_t+D_t)x_{i,t-1}}{A_t}\right]+\E_t\sum_{s=1}^\infty \beta^s[\theta_{i,t+s}u(c_{i,t+s})-\ell_{i,t+s}].
\end{equation*}
Since $(A_t,P_t,D_t,x_{i,t-1})$ is predetermined at time $t$, the continuation utility is ordinally equivalent to
\begin{equation*}
    \E_\theta[\theta u(w_{it}-P_tx_{it})-w_{it}/A_t]+\E_t\sum_{s=1}^\infty \beta^s[\theta_{i,t+s}u(c_{i,t+s})-\ell_{i,t+s}].
\end{equation*}
By the budget constraint \eqref{eq:budget_bewley2}, choosing $\ell_{it}$ is equivalent to choosing $w_{it}$. Because $\theta_{it}$ is \iid across agents and time, all agents face the same problem. Therefore the optimal $w_{it}=w_t$ is common by the strict concavity of $u$.
\end{proof}

The following lemma characterizes the optimal consumption of agents.

\begin{lem}\label{lem:cit}
Let $R_t\coloneqq (P_{t+1}+D_{t+1})/P_t$ be the gross risk-free rate and $w_t$ be the common wealth in Lemma \ref{lem:wt}. Then the optimal consumption is
\begin{equation}
    c_{it}=\min\set{\left(\frac{\theta_{it}A_{t+1}}{\beta R_t}\right)^{1/\gamma},w_t}. \label{eq:cit}
\end{equation}
\end{lem}

\begin{proof}
Using the budget constraint \eqref{eq:budget_bewley2} to eliminate consumption, the continuation utility at time $t$ is
\begin{equation*}
    \E_\theta\left[\theta u(A_t\ell+(P_t+D_t)x_{i,t-1}-P_tx)-\ell\right]+\E_t\sum_{s=1}^\infty \beta^s[\theta_{i,t+s}u(c_{i,t+s})-\ell_{i,t+s}],
\end{equation*}
where $\ell=\ell_{it}$ and $x=x_{it}$. The first-order condition with respect to $\ell$ is
\begin{equation}
    \E_\theta[\theta u'(c_{it})]A_t-1=0. \label{eq:foc_l}
\end{equation}
Noting that $(c_{it},x_{it})$ is chosen after $\theta_{it}$ realizes and the asset cannot be shorted, the first-order condition with respect to $x$ is
\begin{equation}
    -\theta_{it} u'(c_{it})P_t+\beta\E_t[\theta_{i,t+1}u'(c_{i,t+1})](P_{t+1}+D_{t+1})\begin{cases*}
        =0 & if $x_{it}>0$,\\
        \le 0 & if $x_{it}=0$.
    \end{cases*} \label{eq:foc_x}
\end{equation}
Dividing both sides of \eqref{eq:foc_x} by $P_t>0$ and using \eqref{eq:foc_l}, we obtain
\begin{equation}
    -\theta_{it}u'(c_{it})+\frac{\beta R_t}{A_{t+1}}\begin{cases*}
        =0 & if $x_{it}>0$,\\
        \le 0 & if $x_{it}=0$.
    \end{cases*} \label{eq:foc_x2}
\end{equation}
If $x_{it}=0$, the budget constraint implies $c_{it}=w_{it}=w_t$. Otherwise, using \eqref{eq:foc_x2} and $u'(c)=c^{-\gamma}$, we obtain \eqref{eq:cit}.
\end{proof}

\begin{proof}[Proof of Proposition \ref{prop:bewley2}]
If we define $\bar{\theta}_t\coloneqq (\beta R_t/A_{t+1})w_t^\gamma$, then \eqref{eq:cit} immediately implies the optimal consumption rule \eqref{eq:crule}. Using this, we obtain
\begin{equation*}
    \theta u'(c_t(\theta))=\theta c_t(\theta)^{-\gamma}=\frac{\beta R_t}{A_{t+1}}\frac{\theta}{\min\set{\theta,\bar{\theta}_t}}=\frac{\beta R_t}{A_{t+1}}\max\set{1,\theta/\bar{\theta}_t}.
\end{equation*}
Integrating both sides with respect to $F$, using the first-order condition \eqref{eq:foc_l}, and rearranging terms, we obtain \eqref{eq:dynamics_A}.

Finally, let us show \eqref{eq:dynamics_P}. By \eqref{eq:crule}, an agent with $\theta\ge \bar{\theta}_t$ satisfies 
\begin{equation}
    w_t=c_t(\theta)=c_t(\bar{\theta}_t)=\left(\frac{\bar{\theta}_tA_{t+1}}{\beta R_t}\right)^{1/\gamma}. \label{eq:wctheta}
\end{equation}
By Lemma \ref{lem:wt}, $w_t$ is common across all agents, so the budget constraint \eqref{eq:budget_bewley2} implies that labor income is $A_t\ell_{it}=w_t-(P_t+D_t)x_{i,t-1}$. Aggregating across agents, using the market clearing conditions \eqref{eq:cclear} and $\int_I x_{it}\diff i=1$, and using \eqref{eq:crule} and \eqref{eq:wctheta}, we obtain
\begin{align*}
    \int_I c_{it}\diff i&=A_t\int_I \ell_{it}\diff i+D_t=w_t-(P_t+D_t)+D_t\\
    \iff P_t&=w_t-\int_I c_{it}\diff i=c_t(\bar{\theta}_t)-\int_\Theta c_t(\theta)\diff F(\theta)\\
    &=\left(\frac{A_{t+1}}{\beta R_t}\right)^{1/\gamma}\int_\Theta\left(\bar{\theta}_t^{1/\gamma}-\min\set{\theta,\bar{\theta}_t}^{1/\gamma}\right)\diff F(\theta)\\
    &=\left(\frac{A_{t+1}}{\beta R_t}\right)^{1/\gamma}\int_\Theta\max\set{0,\bar{\theta}_t^{1/\gamma}-\theta^{1/\gamma}}\diff F(\theta),
\end{align*}
which is \eqref{eq:dynamics_P}.
\end{proof}

\begin{proof}[Proof of Theorem \ref{thm:bewley2}]
If $\bar{\theta}_t=\theta_L$, then $\theta\ge \theta_L=\bar{\theta}_t$ and \eqref{eq:Ptheta} implies $P_t=0$, which is impossible because $D_t>0$ infinitely often. Therefore $\bar{\theta}_t\ge \theta_L+\delta$ by Assumption \ref{asmp:delta}, and \eqref{eq:Ptheta} implies
\begin{equation}
    P_t\ge A_t^{1/\gamma}((\theta_L+\delta)^{1/\gamma}-\theta_L^{1/\gamma})F(\theta_L)\eqqcolon p A_t^{1/\gamma} \label{eq:Pt_lb}
\end{equation}
for the constant $p>0$. By \eqref{eq:necessity_bewley2}, we can take $\epsilon>0$ such that $D_t^{1/t}<G_d+\epsilon<G-\epsilon<A_t^{1/\gamma t} $ for large enough $t$. Therefore we can bound the dividend yield from above as
\begin{equation*}
    \frac{D_t}{P_t}\le \frac{D_t}{p A_t^{1/\gamma}}\le \frac{1}{p}\left(\frac{G_d+\epsilon}{G-\epsilon}\right)^t,
\end{equation*}
which is summable. By Lemma \ref{lem:bubble}, the equilibrium is bubbly. Furthermore, \eqref{eq:Pt_lb} and the definition of $G$ in \eqref{eq:necessity_bewley2} imply that the equilibrium is asymptotically bubbly.
\end{proof}

\printbibliography

\clearpage

\begin{center}
    {\Huge Online Appendix}
\end{center}

\section{Uniqueness of equilibrium in Example \ref{exmp:linear}}\label{sec:unique}

This appendix characterizes the unique equilibrium in Example \ref{exmp:linear}.

\begin{prop}\label{prop:wilson}
If $1/\beta<G_d<G$ in Example \ref{exmp:linear}, then the unique equilibrium asset price is $P_t=aG^t$, and there is a bubble.
\end{prop}

\begin{proof}
Let $P_t>0$ be any equilibrium asset price. Because the old exit the economy, the equilibrium consumption allocation is
\begin{equation*}
    (y_t,z_t)=(aG^t-P_t,bG^t+P_t+D_t).
\end{equation*}
Nonnegativity of consumption implies $P_t\le aG^t$. Let $R_t\coloneqq (P_{t+1}+D_{t+1})/P_t$ be the gross risk-free rate. The first-order condition for optimality together with the nonnegativity of consumption implies that $R_t\ge 1/\beta$, with equality if $P_t<aG^t$. Suppose $R_t>1/\beta$. Then $P_t=aG^t$, so
\begin{align*}
    R_{t-1}\coloneqq \frac{P_t+D_t}{P_{t-1}}&=\frac{aG^t+DG_d^t}{P_{t-1}}\ge \frac{aG^t+DG_d^t}{aG^{t-1}} && (\because P_{t-1}\le aG^{t-1})\\
    &=\frac{aG^{t+1}+DGG_d^t}{aG^t} > \frac{aG^{t+1}+DG_d^{t+1}}{P_t} && (\because G>G_d, P_t=aG^t)\\
    &\ge \frac{P_{t+1}+D_{t+1}}{P_t}=R_t>\frac{1}{\beta}. && (\because P_{t+1}\le aG^{t+1})
\end{align*}
Therefore by induction, if $R_t>1/\beta$, then $R_s>1/\beta$ for all $s\le t$. This argument shows that, in equilibrium, either
\begin{enumerate*}
    \item\label{item:R=} there exists $T>0$ such that $R_t=1/\beta$ for all $t\ge T$, or
    \item\label{item:R>} $R_t>1/\beta$ for all $t$.
\end{enumerate*}
In Case \ref{item:R=}, using \eqref{eq:Pt}, $1/R_t=\beta$ for $t\ge T$, and $1/\beta<G_d$, the asset price at time $t\ge T$ can be bounded from below as
\begin{equation*}
    P_t\ge V_t=\sum_{s=1}^\infty \beta^s DG_d^{t+s}=\sum_{s=1}^\infty DG_d^t(\beta G_d)^s=\infty,
\end{equation*}
which is impossible in equilibrium. Therefore it must be Case \ref{item:R>} and hence $P_t=aG^t$ and $y_t=0$ for all $t$. In this case, we have
\begin{equation*}
    R_t=\frac{aG^{t+1}+DG_d^{t+1}}{aG^t}\ge G>\frac{1}{\beta},
\end{equation*}
so the first-order condition holds and we have an equilibrium, which is unique. Using $P_t=aG^t$, $D_t=DG_d^t$, and applying Lemma \ref{lem:bubble}, we immediately see that there is a bubble.
\end{proof}

\section{Global dynamics of Example \ref{exmp:CRRA}}\label{sec:global}

This appendix provides a step-by-step analysis of the asymptotically bubbly equilibrium in Example \ref{exmp:CRRA}.

Start with the asset pricing equation \eqref{eq:foc}, which is
\begin{equation}
    P_t=\beta\left(\frac{bG^{t+1}+P_{t+1}+D}{aG^t-P_t}\right)^{-\gamma}(P_{t+1}+D). \label{eq:foc_CRRA}
\end{equation}
Define the detrended variable $\xi=(\xi_{1t},\xi_{2t})$ by $\xi_{1t}\coloneqq P_t/(aG^t)$ and $\xi_{2t}\coloneqq D/(aG^t)$. Then \eqref{eq:foc_CRRA} can be rewritten as the system of autonomous nonlinear implicit difference equations
\begin{equation}
    H(\xi_t,\xi_{t+1})=0, \label{eq:diff_eq}
\end{equation}
where $H:\R^4\to \R^2$ is defined by
\begin{subequations}\label{eq:H}
\begin{align}
    H_1(\xi,\eta)&=\beta G^{1-\gamma}\left(\frac{w+\eta_1+G\xi_2}{1-\xi_1}\right)^{-\gamma}(\eta_1+G\xi_2)-\xi_1, \label{eq:H1}\\
    H_2(\xi,\eta)&=\eta_2-\frac{1}{G}\xi_2 \label{eq:H2}
\end{align}
\end{subequations}
with $w\coloneqq b/a$. Let $\xi^*=(\xi_1^*,\xi_2^*)$ be a steady state of the difference equation \eqref{eq:diff_eq}, so $H(\xi^*,\xi^*)=0$. Since $G>1$ and hence $1/G\in (0,1)$, \eqref{eq:H2} implies that $\xi_2^*=0$. Using \eqref{eq:H1}, we can solve for $\xi_1^*$ as
\begin{equation}
    \beta G^{1-\gamma}\left(\frac{w+\xi_1^*}{1-\xi_1^*}\right)^{-\gamma}\xi_1^*-\xi_1^*=0\iff \xi_1^*=\frac{(\beta G^{1-\gamma})^{1/\gamma}-w}{1+(\beta G^{1-\gamma})^{1/\gamma}}>0 \label{eq:xi1}
\end{equation}
because \eqref{eq:nonexist_CRRA} holds. (Note that the other (fundamental) steady state $\xi_1^*=0$ is ruled out by Theorem \ref{thm:necessity}.)

We apply the implicit function theorem at $(\xi,\eta)=(\xi^*,\xi^*)$ to express \eqref{eq:diff_eq} as $\xi_{t+1}=h(\xi_t)$ for $\xi_t$ close to $\xi^*$. Noting that $\xi_2^*=0$, a straightforward calculation using \eqref{eq:H} and \eqref{eq:xi1} implies that
\begin{equation*}
    D_\xi H(\xi^*,\xi^*)=\begin{bmatrix}
        H_{1,\xi_1} & H_{1,\xi_2}\\
        0 & -1/G
    \end{bmatrix} \quad \text{and} \quad D_\eta H(\xi^*,\xi^*)=\begin{bmatrix}
        H_{1,\eta_1} & 0 \\ 0 & 1
    \end{bmatrix},
\end{equation*}
where
\begin{align*}
    H_{1,\xi_1}&=-\gamma\beta G^{1-\gamma}(w+\xi_1^*)^{-\gamma}(1-\xi_1^*)^{\gamma-1}\xi_1^*-1=-1-\gamma\frac{\xi_1^*}{1-\xi_1^*},\\
    H_{1,\eta_1}&=\beta G^{1-\gamma}\left(\frac{w+\xi_1^*}{1-\xi_1^*}\right)^{-\gamma}\left(1-\gamma\frac{\xi_1^*}{w+\xi_1^*}\right)=1-\gamma\frac{\xi_1^*}{w+\xi_1^*},
\end{align*}
and $H_{1,\xi_2}$ is unimportant. Therefore except the special case with $\gamma=1+w/\xi_1^*$, we may apply the implicit function theorem, and for $(\xi,\eta)$ sufficiently close to $(\xi^*,\xi^*)$, we have $H(\xi,\eta)=0\iff \eta=h(\xi)$ for some $C^1$ function $h$ with
\begin{equation}
    Dh(\xi^*)=\begin{bmatrix}
        \lambda_1 & * \\
        0 & \lambda_2
    \end{bmatrix}, \quad \text{where} \quad (\lambda_1,\lambda_2)=\left(\frac{1+\gamma\frac{\xi_1^*}{1-\xi_1^*}}{1-\gamma\frac{\xi_1^*}{w+\xi_1^*}},\frac{1}{G}\right). \label{eq:Dh}
\end{equation}
We thus obtain the following proposition.

\begin{prop}
Let everything be as in Example \ref{exmp:CRRA} and define $\kappa\coloneqq (\beta G^{1-\gamma})^{1/\gamma}$. If $\frac{1}{\gamma}\neq \frac{\kappa-w}{\kappa(1+w)}$, then there exists an asymptotically bubbly equilibrium such that $P_t/(aG^t)$ converges to $\xi_1^*$ in \eqref{eq:xi1}. If in addition
\begin{equation}
    \frac{1}{\gamma}>\frac{1}{2}\frac{\kappa-w}{\kappa}\frac{1-\kappa}{1+w}, \label{eq:determinacy}
\end{equation}
then such an equilibrium is unique.
\end{prop}

\begin{proof}
By the implicit function theorem, the equilibrium dynamics can be expressed as $\xi_{t+1}=h(\xi_t)$ if $\xi_t$ is sufficiently close to $\xi_1^*$. To study the local stability, we apply the Hartman-Grobman theorem. Since $G>1$, one eigenvalue of $Dh(\xi^*)$ is $\lambda_2=1/G\in (0,1)$. If $1-\gamma\frac{\xi_1^*}{w+\xi_1^*}>0$, then clearly $\lambda_1>1$. If $1-\gamma\frac{\xi_1^*}{w+\xi_1^*}<0$, then
\begin{align*}
    \lambda_1=\frac{1+\gamma\frac{\xi_1^*}{1-\xi_1^*}}{1-\gamma\frac{\xi_1^*}{w+\xi_1^*}}<-1&\iff 1+\gamma\frac{\xi_1^*}{1-\xi_1^*}>-1+\gamma\frac{\xi_1^*}{w+\xi_1^*}\\
    &\iff \frac{1}{\gamma}>\frac{(\kappa-w)(1-\kappa)}{2\kappa(1+w)}
\end{align*}
using the definition of $\kappa$ and $\xi_1^*$. Furthermore, we have
\begin{equation*}
    1-\gamma\frac{\xi_1^*}{w+\xi_1^*}=0\iff \frac{1}{\gamma}=\frac{\kappa-w}{\kappa(1+w)}.
\end{equation*}
Therefore if $\frac{1}{\gamma}\neq \frac{\kappa-w}{\kappa(1+w)}$, the eigenvalues of $Dh(\xi^*)$ are not on the unit circle, so the Hartman-Grobman theorem \citep[Theorem 4.6]{Chicone2006} implies that for sufficiently large $T$ (so that $\xi_{2T}=D/(aG^T)$ is sufficiently close to the steady state value 0), there exists an equilibrium path $\set{\xi_t}_{t=T}^\infty$ starting at $t=T$ converging to $\xi^*$. A backward induction argument similar to Lemma \ref{lem:backward_induction} implies the existence of an equilibrium path $\set{\xi_t}_{t=0}^\infty$ starting at $t=0$. Finally, if \eqref{eq:determinacy} holds, then $\abs{\lambda_1}>1>\lambda_2>0$, so the number of free initial conditions (1, because $P_0$ is endogenous) agrees with the number of unstable eigenvalues (1, because $\abs{\lambda_1}>1$) and the equilibrium path is unique.
\end{proof}

\end{document}